\documentclass[a4paper,leqno]{amsart}

\usepackage{amsmath}
\usepackage{amsthm}
\usepackage{amssymb}
\usepackage{amsfonts}
\usepackage[applemac]{inputenc}
\usepackage{mathrsfs}
\usepackage{pdfsync}
\usepackage{graphicx}
\usepackage{color}

\def\C{{\mathbb C}}
\def\R{{\mathbb R}}
\def\N{{\mathbb N}}
\def\Z{{\mathbb Z}}
\def\T{{\mathbb T}}
\def\Q{{\mathbb Q}}

\def\le{\leqslant}
\def\ge{\geqslant}

\newcommand{\be}{\begin{equation}}
\newcommand{\ee}{\end{equation}}

\newcommand{\re}{\mathrm{Re}}
\newcommand{\im}{\mathrm{Im}}
\newcommand{\e}{\varepsilon}
\newcommand{\supp}{\mathrm{supp}}

\newcommand{\bbR}{{\mathbb R}}

\newcommand{\eps}{\e}

\DeclareMathOperator{\diver}{div}

\theoremstyle{plain}
\newtheorem{theorem}{Theorem}[section]
\newtheorem{lemma}[theorem]{Lemma}
\newtheorem{corollary}[theorem]{Corollary}
\newtheorem{proposition}[theorem]{Proposition}
\newtheorem{hyp}{Assumption}[section]

\theoremstyle{definition}

\newtheorem{remark}[theorem]{Remark}
\newtheorem*{remark*}{Remark}

\numberwithin{equation}{section}

\begin{document}

\title[WKB analysis of Bohmian dynamics]
{WKB analysis of Bohmian dynamics}

\author[A. Figalli]{Alessio Figalli}

\address[A. Figalli]{Department of Mathematics, The University of Texas at Austin\\ 
1 University Station, C1200\\
Austin, USA}
\email{figalli@math.utexas.edu}

\author[C. Klein]{Christian Klein}
 \address[C. Klein]{Institut de Math\'ematiques de Bourgogne \\ 9 avenue Alain Savary\\
21078 Dijon Cedex, France} 
 \email{christian.klein@u-bourgogne.fr}

\author[P. Markowich]{Peter Markowich}
\address[P. Markowich]{King Abdullah University of Science and 
Technology (KAUST)\\ MCSE Division\\
Thuwal 23955-6900\\ Saudi Arabia}
\email{p.markowich@damtp.cam.ac.uk}

\author[C. Sparber]{Christof Sparber}

\address[C.~Sparber]
{Department of Mathematics, Statistics, and Computer Science, M/C 249, University of Illinois at Chicago, 851 S. Morgan Street, Chicago, IL 60607, USA}
\email{sparber@math.uic.edu}

\begin{abstract}
We consider a semi-classically scaled Schr\"odinger equation with WKB initial data. We 
prove that in the classical limit the corresponding Bohmian trajectories converge (locally in measure)
to the classical trajectories before the appearance of the first caustic. In a second step we show that 
after caustic onset this convergence in general no longer holds.
In addition, we provide numerical simulations of the Bohmian trajectories in the semiclassical regime 
which illustrate the above results. \end{abstract}

\date{\today}


\thanks{AF was supported by NSF Grant DMS-0969962.
 CK
 thanks for financial support by  the ANR via the program 
 ANR-09-BLAN-0117-01 and 
 the project FroM-PDE funded by the European
 Research Council through the Advanced Investigator Grant Scheme.}
\maketitle

\section{Introduction}
\label{sec:intro}

\subsection{WKB asymptotics}  We consider the time-evolution of a quantum mechanical particle described by a wave function
$\psi^\e (t, \cdot) \in L^2(\R^d; \C)$ and governed by the Schr\"odinger equation (in dimensionless form):
\begin{equation}
\label{sch}
i\e \partial_t  \psi^\e  = -\frac{\e^2}{2}\Delta \psi^\e +
V(x)\psi^\e,\quad
\psi^\e|_{t=0}   = \psi^\e_{0} ,
\end{equation}
where $x \in \R^d$, $t\in \R_+ $,  and $V(x)\in \R$ a given potential (satisfying some regularity assumptions to be specified below).
In addition, we have rescaled the equation such that only one 
\emph{semi-classical parameter} $0<\e \ll1$ remains. 

The \emph{classical limit} of quantum mechanics 
is concerned with the asymptotic behavior of solutions to \eqref{sch} as $\e \to 0_+$. 
A possible way to describe these asymptotics is based on the time-dependent \emph{WKB method}, where 
one makes the ansatz (see, e.g., \cite{Ca, Ca1} for more details)
\begin{equation}\label{wkbansatz}
\psi^\e(t,x) = a^\e(t,x) e^{i S (t,x)/\e}
\end{equation}
for some $\e$-independent (real-valued) phase function $S(t,x)\in \R$ and a (in general complex valued) amplitude $a^\e(t,x)\in \C$ satisfying 
$$
a^\e \sim a + \e a_1 + \e^2 a_2+\dots,
$$
in the sense of asymptotic expansions. Assuming for the moment that $a^\e$ and $S$ are sufficiently smooth, one can plug \eqref{wkbansatz} into \eqref{sch} and compare
equal powers of $\e$ in the resulting expression. This yields a \emph{Hamilton-Jacobi equation} for the phase
\begin{equation}\label{hj}
\partial_t S + \frac{1}{2} |\nabla S|^2 + V(x)= 0, \quad S|_{t=0} = S_0,
\end{equation}
and a transport equation for the leading order amplitude
\begin{equation}\label{wkbamp}
\partial_t a + \nabla a \cdot \nabla S + \frac{a}{2} \Delta S= 0,\quad a|_{t=0} = a_0.
\end{equation}
Note that the latter can be rewritten in the form of a conservation law for the leading order particle density $\rho:=|a|^2$,
i.e., 
\begin{equation}\label{wkbcons}
\partial_t \rho + \diver (\rho \nabla S) = 0.
\end{equation}
The main problem of the WKB approach is that \eqref{hj} in general does not admit unique smooth solutions for all times. This can be seen, from the method of characteristics, 
where one needs to solve the following Hamiltonian system
\begin{equation}\label{classflow}
\left \{
\begin{aligned}
& \,  \dot X(t,y)= P(t,y) , \quad X(0,y)=y,\\
& \, \dot P(t,y)= - \nabla V(X(t,y)) , \quad P(0,y)=\nabla S_0(y).
\end{aligned}
\right. 
\end{equation}
By the Cauchy-Lipschitz theorem, this system of ordinary differential equations can be solved at least locally
in-time, which yields a flow map $X_t: y\mapsto  X(t,y)$. 
If we denote the corresponding inverse mapping by $Y_t: x \mapsto Y(t,x)$, i.e., $Y_t \circ X_t = {\rm id}$,
then the phase function $S$ satisfying \eqref{hj} is found to be (see, e.g., \cite{Ca1})
\begin{equation}\label{phi}
S(t,x) = S_0(Y(t,x)) + \int_0^t \biggl(\frac{1}{2} |P(\tau, y)|^2 - V(X(\tau, y))\biggr)\, d \tau\big|_{y=Y(t,x)}.
\end{equation}
Given such a smooth phase function $S$, one can, in a second step, integrate the amplitude equation \eqref{wkbamp} along the flow $X_t$ to obtain 
the amplitude in the form 
\begin{equation}\label{a}
a(t,x) = \frac{a_0(Y(t,x))}{\sqrt{J_t(Y(t,x))}}\, , 
\end{equation}
where $J_t(y):=\text{det} \nabla_y X(t,y)$ is the Jacobian determinant of the map $y\mapsto X(t,y)$.
The problem is that in general there is a (possibly, very short) time $T^*>0$, at which the flow $X_t$ ceases to be one-to-one. 
Points $x \in \R^d$ at which this happens are \emph{caustic points} and $T^*$ is called the 
{\it caustic onset time}. More precisely, let
\begin{equation*}\label{Ct}
 \mathscr{ C}_t = \{x \in \R^d \, : \, \text{there is $y\in \R^d$ such that $x = X(t,y)$ and $J_t(y) = 0$} \},
\end{equation*}
then the {\it caustic set} is defined by $ \mathscr  C:= \{(x,t) : x\in  \mathscr  C_t \}$ and the caustic onset time is $$T^* := \inf \{ t\in \R_+ \, : \, C_t \not = \emptyset \}.$$ 
For $t>T^*$ the solution of \eqref{hj}, obtained by the method of characteristics, typically becomes 
multi-valued due to the possibility of crossing trajectories, see Fig.~\ref{caustic}.
\begin{figure}[htb]
    \begin{center}
      \includegraphics[width=0.6\textwidth]{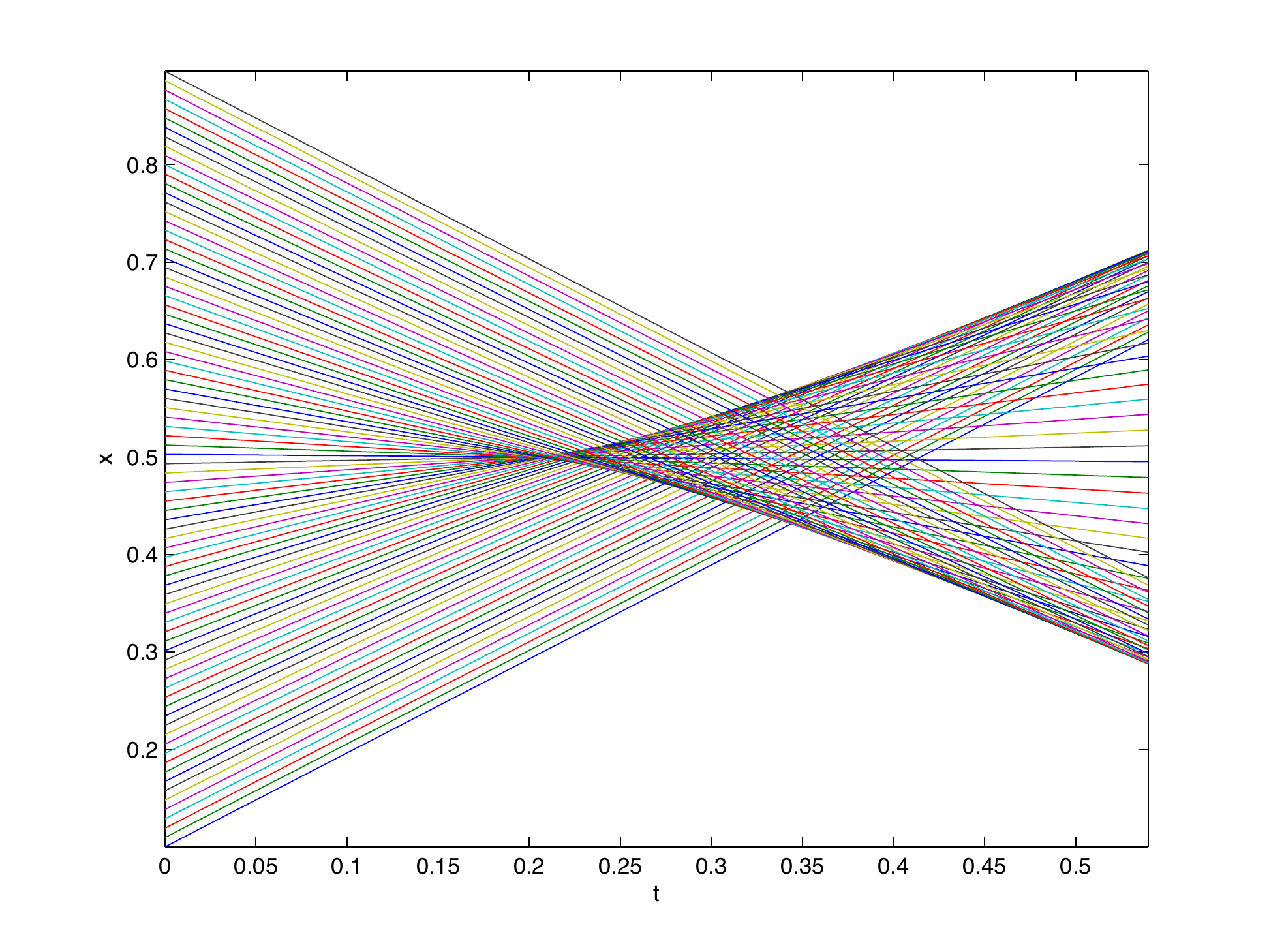}
    \end{center}
\caption{Classical trajectories for initial data $\nabla S_0(x)= - 
\tanh(5x-\frac{5}{2})$}
\label{caustic}
\end{figure}
On the other hand, weak solutions to \eqref{hj}, which can be uniquely defined (for example, by invoking the Lax-Olejnik formula) 
are not smooth in general and thus plugging \eqref{wkbansatz} into \eqref{sch} is no longer justified. From the physical point of view $T^*$ marks the generation of new frequencies
within $\psi^\e$ not captured by the simple one phase WKB ansatz \eqref{wkbansatz}.
Indeed, it is well known that for $t>T^*$ one generically 
requires a {\it multi-phase} WKB ansatz to correctly describe the asymptotic behavior of $\psi^\e$, see Section \ref{sec:free} for more details.

\subsection{Bohmian trajectories and Bohmian measures} It is well known, that to any wave function $\psi^\e\in H^1(\R^d)$ one can associate two basic observable densities. Namely, the \emph{position} 
and the \emph{current-density}, defined by
\begin{equation}\label{densities}
\rho^\e(t,x)= |\psi^\e(t,x)|^2, \quad J^\e(t,x) = \e \im\big(\overline{\psi^\e}(t,x)\nabla \psi^\e(t,x)\big),
\end{equation}
which satisfy the conservation law
\begin{equation*}
\partial_t \rho^\e + \diver_x J^\e = 0.
\end{equation*}
These two quantities play an important role in {\it Bohmian mechanics} developed in \cite{Bo1, Bo2} (see also \cite{DuTe} for a general introduction). 
In this theory, one defines $\e$-dependent 
trajectories $X^\e_t: y \mapsto X^\e(t,y)$, via the following differential equation
\begin{equation*}\label{bohm1}
\dot X^\e(t, y) = u^\e(t,X^\e(t, y)) ,\quad X^\e(0, y) = y\in \R^d,
\end{equation*} 
where the initial data $y\in \R^d$ is assumed to be distributed according to the measure $\rho_0^\e \equiv |\psi^\e_0|^2\in L_+^1(\R^d)$ and 
$u^\e$ denotes the quantum mechanical velocity field, (formally) defined by
\begin{equation}\label{velocity}
u^\e(t,x):= \frac{J^\e(t,x)}{\rho^\e(t,x)}= \e \im \left(\frac{ \nabla \psi^\e(t, x)}{ \psi^\e(t, x)} \right ).
\end{equation}
Note that 
the kinetic energy of $\psi^\e$ can be written in terms of $\rho^\e$ and $u^\e $ as
\begin{equation}\label{Ekin}
\begin{split}
E_{kin}(t): = & \ \frac{\e^2}{2} \int _{\R^d} | \nabla \psi^\e (t,x) |^2 dx \\
=& \ \frac{1}{2}  \int _{\R^d} \rho^\e(t,x) |u^\e(t,x)|^2 \, dx +  \frac{\e^2}{2} \int _{\R^d}  | \nabla \sqrt{\rho^\e}(t,x) |^2 dx, 
\end{split}
\end{equation}
which allows to define $u^\e\in L^2(\R^d; \rho^\e dx)$.
Although $u^\e$ is {\it not} necessarily continuous,
it was rigorously proved in \cite{BDGPZ, TeTu} that, for all $t\in \R$, $x\mapsto X^\e(t,\cdot)$ is well-defined $\rho_0^\e-a.e.$. In addition, 
one finds (under mild assumptions on the potential $V$) that, for all times $t\in \R_+$, 
the position density $\rho^\e(t,x)$ is given by the \emph{push-forward} of the initial density $ \rho^\e_{0} (x)$ under the mapping
$X^\e_t$, i.e., for any non-negative Borel function $\sigma:\bbR^d\to [0,+\infty]$ it holds: 
\begin{equation}\label{push1}
\int_{\bbR^d}\sigma(x)\rho^\e(t,x) dx=\int_{\bbR^d}\sigma(X^\e(t,y))\rho^\e_0(y) dy.
\end{equation}
In \cite{MPS1}, a phase space description of Bohmian mechanics was rigorously introduced through the definition of a
class of positive measures, called {\it Bohmian measures}, $\beta^\e \in \mathcal M^+(\R^d_x\times \R^d_p)$,
associated to $\rho^\e$ and $J^\e$. (Here and in the following, $\mathcal M^+$ denotes the set of 
non-negative Borel measures.)
Indeed, for a given scale $\e>0$ and any $\psi^\e \in H^1(\R^d)$, one defines the associated Bohmian measure $\beta^\e$ by
\[
\beta^\e(t,x,p):= \rho^\e(t,x) \delta (p - u^\e(t,x) ),
\]
where $u^\e$ is defined by \eqref{velocity} and $\delta(p-\, \cdot)$ denotes the $d$-dimensional 
delta distribution with respect to the momentum variable $p\in \R^d$. Note that even though $u^\e$ is not well defined at points where $\rho^\e(t,x)=0$, 
the Bohmian measure $\beta^\e$ is.
In addition, the zeroth and first moment of $\beta^\e$ with respect to $p\in \R^d$ 
yield the quantum mechanical particle and current densities, i.e., 
\[
\rho^\e(t,x) = \int_{\R^d} \beta^\e(t, x, dp), \quad J^\e(t,x) = \int_{\R^d} p \beta^\e(t, x, dp).
\]
It is shown in \cite[Lemma 2.5]{MPS1}
that, for all $t\in \R_+$, the measure $\beta^\e(t,x,p) $ is given by the push-forward of the initial measure 
$$\beta^\e_0(y,p)=\rho^\e_0 (y) \delta(p - u^\e_0(y)),$$ 
under the following $\e$-dependent flow:
\begin{equation}\label{bohmflow}
\left \{
\begin{aligned}
& \,  \dot X^\e(t,y)= P^\e(t,y) , \quad X^\e(0,y)=y,\\
& \, \dot P^\e(t,y)= - \nabla V(X^\e(t,y)) -  \nabla V_{B}^\e(t,X^\e(t,y)), \quad P^\e(0,y)= u^\e_0(y),
\end{aligned}
\right. 
\end{equation}
where $V^\e_B(t,x)$, denotes the \emph{Bohm potential}:
\begin{equation}
\label{bohmpot}
V^\e_B:= -\frac{\e^2}{2\sqrt{\rho^\e}}\, \Delta \sqrt{\rho^\e}.
\end{equation}
More precisely, for any non-negative Borel function $\varphi :\R^d_x\times\R^d_p\to [0,+\infty]$
it holds 
\begin{equation}\label{push2}
\iint_{\R_{x,p}^{2d}} \varphi(x,p) \beta^\e (t, dx, dp) =  \int_{\R_y^{d}} \varphi(X^\e(t,y),P^\e(t,y))\rho_0^\e(y) \, dy .
\end{equation}
Note that \eqref{bohmflow} is the characteristic flow of the 
following perturbed Burgers' type equation
\[
\partial_t u^\e+ \left( u^\e \cdot \nabla \right) u^\e + \nabla V = \nabla V^\e_B(t,x),\quad u^\e|_{t=0}= u^\e_0,
\]
which allows us to identify $ \dot X^\e(t,y) =P^\e(t,y)  = u^\e(t, X^\e(t,y))$. 
On the other hand, for wave functions $\psi^\e$ given in WKB form \eqref{wkbansatz}, we have $J^\e = |a^\e|^2 \nabla S$ in which case the velocity field is simply given by
$$u^\e(t,x)=\frac{J^\e(t,x)}{\rho^\e(t,x)} = \nabla S(t,x).$$ 
One can therefore regard \eqref{bohmflow} as a nonlinear perturbation 
of the classical equations of motion \eqref{classflow}. One consequently expects the Bohmian trajectories $(X^\e, P^\e)$ to converge to the corresponding classical $(X, P)$,
in the limit $\e \to 0_+$. We shall prove that, at least before caustic onset, this convergence indeed holds true (in a sense to be made precise, see Theorem \ref{th1}). 
After caustic onset, however, the situation in general is much more complicated in view of Fig.~\ref{caustic}.
Indeed, we shall show that in general one cannot expect the Bohmian trajectories to 
converge to the (multi-valued) classical flow, see Theorem \ref{thm:free}.

The rest of the paper is organized as follows:
In Section \ref{sec:apriori} we describe some general properties of Bohmian dynamics and of the Young measures
associated to the Bohmian trajectories. These properties will be used in Section \ref{sec:pre} to prove that 
the Bohmian trajectories converge to the classical ones before caustic onset.
In Section \ref{sec:super} we prove a general result about Bohmian measures associated to multi-phase WKB states.
This result is then used in Section \ref{sec:free} to show that, even in the free case (where $V(x)\equiv0$), the Bohmian
measure may differ from the Wigner measure, and that in general
the Bohmian trajectories do not converge to the Hamiltonian ones after caustics.
Finally, in Section \ref{sec:num} we present a numerical simulations of Bohmian trajectories in the
regime $0<\e\ll 1$.

\section{Mathematical preliminaries}\label{sec:apriori}

\subsection{Basic a-priori estimates and existence of a limiting measure} From now on the potential $V$ will satisfy the following assumptions.
\begin{hyp}\label{hypV} 
The potential $V\in C^\infty(\R^d;\R)$ is assumed to be bounded below and sub-quadratic, i.e.,
$$
\partial_x^k V \in L^\infty(\R^d)\, ,\quad \forall k\in \N^d \text{
such that }|k|\ge 2. 
$$
\end{hyp}
Since $V$ is bounded below, without loss of generality we can assume $V(x)\ge 0$. 
Assumption \ref{hypV} is (by far) sufficient to guarantee the existence of a unique strong solution $\psi^\e \in C(\R^d; L^2(\R^d))$ 
to \eqref{sch}, satisfying two basic conservation laws of quantum mechanics. Namely, conservation of the total {\it mass}
\begin{equation}\label{mass}
M^\e(t):=\int _{\R^d} | \psi^\e (t,x) |^2 dx = M^\e(0) ,
\end{equation}
and the {\it total energy}
\begin{equation}\label{energy}
E^\e(t) : =   \frac{\e^2}{2} \int _{\R^d} | \nabla \psi^\e (t,x) |^2 dx + \int_{\R^d} V(x) |  \psi^\e (t,x) |^2 dx =E^\e(0).
\end{equation}
A direct consequence of these conservation laws is the following result to be used later on.
\begin{lemma} \label{lem:apriori} Let $V$ satisfy Assumption \ref{hypV} and $\psi_0^\e \in H^1(\R^d)$. Then, it holds: 
\[ 
\int_0^T \int_{\R^d} | P^\e(t,y)|^2 \rho^\e_0 (y) \, dy \, dt \le T E^\e(0), \quad  \forall T \in \R_+ .
\]
\end{lemma}
\begin{proof} Let us recall that $\rho^\e(t,x)$ is the push forward of $\rho_0^\e$ under the mapping $X^\e_t$, i.e., 
identity \eqref{push1} holds true for all $t\in \R_+$. Using this identity with $ \sigma(\cdot) = |P^\e(t,\cdot)|^2$
and recalling that $P^\e(t,y) = \dot X^\e(t,y) = u^\e(t, X^\e(t,y))$, we find
\begin{align*}
 \int_0^T \int_{\R^d} | P^\e(t,y)|^2 \rho^\e_0 (y) \, dy \, dt 
  & = \int_0^T \int_{\R^d} | u^\e(t,X^\e(t,y))|^2 \rho^\e_0 ( y) \, dy \, dt\\ 
  & = \int_0^T \int_{\R^d} | u^\e(t,y)|^2 \rho^\e (t, y) \, dy \, dt.
\end{align*}
In view of energy conservation, the last term on the right hand side is bounded by
\begin{align*}
\int_0^T \int_{\R^d} | u^\e(t,y)|^2 \rho^\e (t, y) \, dy \, dt  \le   \int_0^T E^\e(t) \, dt = T E^\e(0),
\end{align*}
as desired.
\end{proof}
In addition, to Assumption \ref{hypV} we require the following basic properties for the initial datum $\psi^\e_0$.
\begin{hyp}\label{hypini} 
The initial data of \eqref{sch} satisfy $M^\e(0) \equiv \| \psi_0^\e\|_{L^2}^2=1,$
and there exists $C_0>0$ such that 
\begin{equation*}\label{energybound}
\sup_{0<\e \le 1} E^\e(0) \le C_0.
\end{equation*}
\end{hyp}
\begin{remark} The normalization $\| \psi_0^\e\|_{L^2}^2=1$ is imposed for the sake of mathematical convenience. 
>>From a physical point of view, it is required for the usual probabilistic interpretation of quantum mechanics in which $\rho^\e = |\psi^\e|^2$ denotes the probability measure of finding 
the particle within a certain spatial region  $\Omega\subseteq \R^d$. \end{remark}

Assumption \ref{hypini}, together with conservation of mass and energy and the fact that $V(x)\ge0$, implies that for all $ t\in \R_+$:
\begin{equation}\label{epsosc}
\sup_{0<\e \le 1}( \| \psi^\e(t)  \|_{L^2} +   \| \e \nabla \psi^\e(t)  \|_{L^2})< + \infty.
\end{equation}
In other words, $\psi^\e(t)$ is {\it $\e$--oscillatory} and we are in the framework of \cite{MPS1}.
Indeed, it was shown in  \cite[Lemma 3.1]{MPS1} that 
\eqref{epsosc} implies the existence of a limiting measure $\beta (t)\in \mathcal M^+(\R^d_x\times \R^d_p)$ such 
that, up to extraction of a subsequence, it holds:
\begin{equation}\label{bohmlim}
\beta^\e  \stackrel{\e\rightarrow 0_+ }{\longrightarrow} \beta, \quad \text{in
$L^\infty(\R_t;\mathcal M^+(\R_x^d\times \R^d_p)) \, {\rm weak-}\ast$,}
\end{equation}
and we also have
\begin{equation}\label{Bmoment}
\rho^\e(t,x) \stackrel{\e\rightarrow 0_+ }{\longrightarrow} \int_{\R^d} \beta(t, x, dp),
\quad J^\e(t,x) \stackrel{\e\rightarrow 0_+ }{\longrightarrow} \int_{\R^d} p \beta(t, x, dp),
\end{equation}
where the limits have to be understood in $L^\infty(\R_t; \mathcal M^+(\R^d_x)) \, {\rm weak-}\ast$.

\subsection{Young measures of Bohmian trajectories} The limiting Bohmian measure $\beta$ is intrinsically connected to the 
Young measure (or parametrized measure) of the Bohmian dynamics. 
To this end, we first note that $\Phi^\e(t,y)\equiv (X^\e(t,y), P^\e(t,y))$ is measurable in $t,y$ and 
thus, there exists an associated Young measure $$ \Upsilon_{t,y}: \R_t \times \R^d_y \to \mathcal M^+(\R^d_y \times \R^d_p): \quad (t,y) \mapsto \Upsilon_{t,y} (dx, dp),$$
which is defined through the following limit (see \cite{Ba, Hu, Pe1}):
for any test function $\sigma \in L^1(\R_t\times \R^d_y; C_0(\R^{2d}))$,
$$
\lim_{\e \to 0} \iint_{\R\times \R^d} \sigma ( t, y, \Phi_\e(t,y)) \, dy \, dt = \iint_{\R \times \R^d} \iint_{\R^{2d}} \sigma ( t,y,x,p) \Upsilon_{t,y}(dx,dp) \, dy\,  dt.
$$
Having in mind 
\eqref{push2}, if we assume in addition that 
\[\rho_0^\e  \stackrel{\e\rightarrow 0_+ }{\longrightarrow} \rho_0, \quad \text{strongly in $L^1_+(\R^d)$,}\]
we easily get
the following identity:
\begin{equation}\label{formula}
\beta(t,x,p) = \int_{\R_y^d}  \Upsilon_{t,y} (x,p )  \rho_0(y) dy .
\end{equation}
Here, $\beta$ is the limiting Bohmian measure obtained in \eqref{bohmlim} for a specific subsequence.
The relation \eqref{formula} has already been observed in \cite{MPS1} and can be used to infer the following a-priori estimate on $\Upsilon_{t,y}$.
\begin{lemma}\label{lem:apriori2} 
Let Assumptions \ref{hypV} and \ref{hypini} hold, and assume in addition that
$\rho_0^\e  \stackrel{\e\rightarrow 0_+ }{\longrightarrow} \rho_0$ strongly in $L^1_+(\R^d)$. Then, for any $T\in \R_+$, 
there exists a $C=C(T)>0$ such that
\[
\int_0^T \iiint_{ \R^{2d}\times \R^d} |p|^2 \rho_0(y) \Upsilon_{t,y}(dx,dp)  \, dy\,  dt \le C(T).
\]
\end{lemma}
\begin{proof} Using \eqref{formula} we see that 
\[
\iiint_{ \R^{2d} \times \R^d} |p|^2\rho_0(y) \Upsilon_{t,y}(dx,dp) \, dy =  \iint_{ \R^{2d}} |p|^2 \beta(t,dx,dp).
\]
Now we recall that, by definition, \[\beta^\e(t,x,p) = \rho^\e(t,x) \delta(p-u^\e(t,x))\] and hence
\[
\iint_{ \R^{2d}} |p|^2 \beta^\e(t,dx,dp) = \int _{\R^d} \rho^\e(t,x) |u^\e(t,x)|^2 \, dx \le 2E^\e_{\rm kin}(t) \le C(T),
\]
in view of \eqref{Ekin} and energy conservation. This uniform (in $\e$) bound together with Fatou's lemma implies
\[
 \iint_{ \R^{2d}} |p|^2 \beta(t,dx,dp)\le C(T),
 \]
 and the assertion is proved.
\end{proof}
Lemma \ref{lem:apriori2} together with Lemma \ref{lem:apriori} will be used to prove the following important property for the zeroth moment of $\Upsilon_{t,y}$.

\begin{proposition} \label{proptrans} Let Assumptions \ref{hypV} and \ref{hypini} hold, and assume in addition that $\rho_0^\e  \stackrel{\e\rightarrow 0_+ }{\longrightarrow} \rho_0$ strongly in $L^1_+(\R^d)$. Denote 
\[
\upsilon_{t,y}(x):= \int_{\R^d} \Upsilon_{t,y}(x,dp).\]
Then $\upsilon_{t,y}\in \mathcal M^+(\R_x^d)$ solves, a.e. with respect to the measure $\rho_0(y)$, the following transport equation
\[
\partial_ t \upsilon_{t,y} + \diver_x \left( \int_{\R^d} p \Upsilon_{t,y}(x,dp) \right) = 0, \quad \upsilon_{t=0,y}(x) = \delta(x-y),
\]
in the sense of distributions on $\mathcal D'(\R_t\times \R^d_x)$.
\end{proposition}
This transport equation will play a crucial role in the convergence proof of Bohmian trajectories before caustic onset.
\begin{proof} As a first, preparatory step we shall prove that, for all test functions $\zeta \in C_0(\R_t\times \R^d_y)$, $\sigma \in C_0( \R^d_x)$:
\begin{equation}\label{techlim}
\begin{split} 
\lim_{\e \to 0_+} \int_0^T \int_{\R^d}   P^\e(t,y) \zeta(t,y) \sigma(X^\e(t,y))  \rho^\e_0 (y) \, dy \, dt =  \\
\int_0^T \zeta(t,y)\iint_{\R^{2d}} p \, \sigma(x)    \Upsilon_{t,y}(dx, p) \, \rho_0 (y) \, dy \, dt ,
\end{split}
\end{equation} 
To this end, let $K>0$ and $\chi_K\in C^\infty_c(\R^d)$ be such that and $\chi_K(p) = 1$ for $|p| \le K$, and $\chi_K(p)=0$ for $|p|>K+1$.
Then, by writing $P^\e = \chi_K(P^\e) + (1-\chi_K(P^\e))$ we can decompose
\[ \int_0^T \int_{ \R^d}  \zeta(t,y) \sigma(X^\e(t,y))  P^\e(t,y) \rho^\e_0 (y) \, dy \, dt  = I^{\e, K}_1 + I^{\e, K}_2.\]
Because of the strong convergence of $\rho_0^\e$, the first term on the right hand side has the following limit:
\[
I^{\e, K}_1 \stackrel{\e \rightarrow 0_+ }{\longrightarrow}  \int_0^T \zeta(t,y)\iint_{\R^{2d}} \sigma(x)  \chi_K(p)   \Upsilon_{t,y}(dx, dp) \, \rho_0 (y) \, dy \, dt ,
\]
On the other hand, by having in mind the result of Lemma \ref{lem:apriori}, the second term on the right hand side can be estimated by 
\begin{align*}
|I^{\e, K}_2| & \le C \int_0^T \int_{|P^\e|\ge K} | P^\e(t,y)| \rho^\e_0(y) \, dy, \, dt \\
& \le \frac{C}{K} \int_0^T \int_{\R^d}| P^\e(t,y)|^2 \rho^\e_0(y) \, dy, \, dt \le \frac{CT}{K} E^\e(0).
\end{align*}
In view of Lemma \ref{lem:apriori2} we can let $K\to +\infty$, which yields
$|I^{\e, K}_2| \stackrel{K\rightarrow +\infty }{\longrightarrow}0$ and the validity of \eqref{techlim}.

With \eqref{techlim} in hand, we shall now show that $\upsilon_{t,y}$ indeed obeys the transport equation given above. 
Let $\zeta, \varphi \in C_{\rm c}^\infty(\R^d)$, $\sigma \in C_{\rm c}^\infty[0,\infty)$,
be smooth compactly supported test functions. Then by \eqref{techlim} we get
\begin{align*}
&\int_0^\infty \iint_{\R^{2d}}  \Bigl(\partial_t \sigma(t) \varphi(x)   + \sigma(t) p\cdot \nabla_x \varphi(x) \zeta(y) \Bigr) \Upsilon_{y,t}(x,dp) \rho_0(y)  \, dy \, dt\\
&= \lim_{\e \to 0_+} \int_0^\infty \int_{\R^{d}}
\Bigl(\partial_t \sigma(t) \varphi(X^\e(t,y))   +
\sigma(t) P^\e(t,y) \cdot \nabla_x \varphi(X^\e(t,y)) \zeta(y) \Bigr)  \rho_0(dy)  dt.
\end{align*}
Recalling that $P^\e(t,y) = \dot X^\e(t,y)$, which implies that
\[
P^\e(t,y) \cdot \nabla_x \varphi(X^\e(t,y)) = \frac{d}{dt} \varphi(X^\e(t,y)),
\]
we obtain 
\begin{align*}
& \int_0^\infty \int_{\R^{d}} \Bigl(\partial_t \sigma(t) \varphi(X^\e(t,y))
 + \sigma(t) P^\e(t,y) \cdot \nabla_x \varphi(X^\e(t,y)) \zeta(y) \Bigr)  \rho_0(dy)  \, dt \\
& = \int_0^\infty \int_{\R^{d}} \Bigl(\partial_t \sigma(t) \varphi(X^\e(t,y))
+  \sigma(t) \frac{d}{dt} \varphi(X^\e(t,y)) \zeta(y) \Bigr)  \rho_0(y)\,  dy\, dt\\
& =  \int_{\R^{d}} \sigma(0)\varphi(X^\e(0,y))  \zeta(y)   \rho_0(y) \, dy \\
& =  \int_{\R^{d}}  \sigma(0) \varphi(y)  \zeta(y)   \rho_0(y) \, dy.
\end{align*}
where in going from the second to the third we have itegrated by parts with respect to time,
and from the third to the forth line 
we have used that $X^\e(0,y) = y$ by definition. The obtained expression in the last line is nothing but the initial condition, since
\[
 \iint_{\R^{d}}  \varphi(y)  \zeta(y)   \rho_0(y)  dy =  \iint_{\R^{2d}}   \varphi(x)  \Upsilon_{0,y} (x,dp)\zeta(y)   \rho_0(y)  dy,
\]
is equivalent to saying that $$\upsilon_{0,y}(x) \equiv \int_{\R^d} \Upsilon_{0,y} (x,dp) = \delta(x-y), \quad \rho_0(dy) - a.e.$$
\end{proof}

Having collected all necessary properties of $\Upsilon_{t,y}$ we shall prove the convergence of Bohmian trajectories (before caustic onset) in the next section.

\begin{remark} For completeness, we want to mention that $\Upsilon_{t,y}$ is indeed a probability measure
on $\R^d_x\times \R^d_p$ for a.e. $y, t$, provided 
the sequence $\{\psi^\e\}_{0<\e \le 1}$ is compact at infinity (tight), i.e.,
$$
\lim_{R\to \infty} \limsup_{\e \to 0_+} \int_{|x|\ge R}|\psi^\e(t, x)|^2\, dx = 0 . 
$$
Indeed if the latter holds true, it was shown in \cite[Lemma 3.2]{MPS1} that
\[
\lim_{\e \to 0_+}M^\e(t) \equiv \lim_{\e \to 0_+} \iint_{\R^{2d}} \beta^\e (t, dx, dp) = \iint_{\R^{2d}} \beta(t, dx, dp) ,
\]
and having in mind our normalization $M^\e(t)=1$, we conclude
$$
1 = \iint_{\R^{2d}} \beta(t, dx, dp) =  \iiint_{\R^{2d} \times {\R^d}}  \rho_0(y) \Upsilon_{t,y} (dx,dp ) \, dy,
$$
in view of \eqref{formula}. Define
$$
\alpha_{t,y} :=\iint_{\R^{2d}}  \Upsilon_{t,y} (dx,dp ) \le 1.
$$
Then, since $\int_{\R^d} \rho_0 (dy) =1$, we conclude $\alpha_{y,t} =1$ a.e.. However, we shall not use this property in the following.
\end{remark}

\section{Convergence of Bohmian trajectories before caustic onset}\label{sec:pre}

So far we have not specified the initial data $\psi_0^\e$ to be of WKB form.
By doing so, we can state the first main result of our work (recall the definition of sub-quadratic,
given in Assumption \ref{hypV}).

\begin{theorem} \label{th1} Let Assumptions \ref{hypV} hold, and let $\psi_0^\e$ be given in WKB form
\begin{equation}\label{WKBini1}
\psi^\e_0(x)= a_0(x) e^{i S_0(x)/\e},
\end{equation}
with amplitude $a_0\in \mathcal S(\R^d;\C)$ and sub-quadratic phase $S_0 \in C^\infty (\R^d;\R)$. 
Then, there exists a caustic onset time $T^*>0$ such that:

\emph{(i)} For all compact time-intervals $I_t\subset [0, T^*)$, the Bohmian measure $\beta^\e$ associated to $\rho^\e, J^\e$ satisfies
$$
\beta^\e  \stackrel{\e\rightarrow 0_+ }{\longrightarrow} \rho(t,x) \delta(p -\nabla S(t,x)), \ \text{in $L^\infty(I_t;\mathcal M^+(\R_x^d\times \R^d_p)) \, {\rm weak-}\ast$,}
$$
where $\rho\in C^\infty(I_t;\mathcal S(\R^d))$ and $S \in C^\infty (I_t\times \R^d)$ solve the WKB system \eqref{wkbcons}, \eqref{hj}.

\emph{(ii)} The corresponding Bohmian trajectories satisfy
$$
X^\e  \stackrel{\e\rightarrow 0_+ }{\longrightarrow} X, \quad P^\e  \stackrel{\e\rightarrow 0_+ }{\longrightarrow} P 
$$
locally in measure on $\{ I_t\times \supp \,\rho_0\}\subseteq \R_t\times \R^d_x$, where $\rho_0 = |a_0|^2$.
More precisely, for every $\delta>0$ and every 
Borel set $\Omega \subseteq \{ I_t \times \supp \, \rho_0\}$ with finite Lebesgue measure $\mathscr L^{d+1}$, it holds
\[
\lim_{\e\to 0} \mathscr L^{d+1} \big(\{(t,y)\in \Omega:\ | (X^\e (t,y),P^\e (t,y))-(X(t,y),P(t,y))| \ge \delta \} \big) = 0.
\]
\end{theorem}

Assertion (i) above was already proved in \cite{MPS1}, but since the obtained form of the limiting measure will be used to show (ii), 
we shall recall the proof of (i) for the sake of completeness. 
Assertion (ii) shows, that before caustic onset, the Bohmian trajectories converge locally in measure to the corresponding classical flow. 
Clearly, if $a_0(x)>0$ for all $x \in \R^d$, and thus $\supp \, \rho_0 =\R^d$, we obtain local in measure convergence of the 
Bohmian trajectories on all of $I_t \times \R^d_x$. After selecting an appropriate sub-sequence $\{ \e_n \}_{n\in \N}$ this also implies (see, e.g., \cite{Bo}) {\it almost everywhere convergence} on any finite 
subset of $I_t \times \R^d_x$.
Moreover, since, by definition, $\dot X^\e=P^\e$, the convergence in measure of $P^\e$ to $P$ combined with the $L^2$
bound from Lemma \ref{lem:apriori} implies that, for $\mathscr L^d$-a.e. $y$, the curves $X^\e(\cdot, y)$ converge
{\it uniformly} to $X(\cdot, y)$ on the time interval $I_t$.

\begin{proof}[Proof of Theorem \ref{th1}] We first note that \eqref{WKBini1} implies 
\[E^\e(0) =   \frac{1}{2} \int _{\R^d} |a_0|^2 | \nabla S |^2 dx + \frac{\e^2}{2} \int _{\R^d} |\nabla a_0|^2 dx + \int_{\R^d} V(x) |  a_0 |^2 dx. \]
Since $a_0 \in \mathcal S(\R^d)$, we see that Assumption \ref{hypini} is satisfied and
thus all the results established in Section \ref{sec:apriori} apply. In particular, we have the 
existence of a limiting Bohmian measure
$\beta\in L^\infty(\R_t; \mathcal M^+(\R^d_x\times \R^d_p)) \, {\rm weak-}\ast$. In order to prove Assertion (i) we need to show that before caustic onset, this limiting measure is 
given by a {\it mono-kinetic} phase space distribution, i.e.,
\be 
\label{mono} 
\beta(t,x,p)= \rho(t,x) \delta(p -\nabla S(t,x)).
\ee
In \cite{MPS1} sufficient conditions for $\beta$ being mono-kinetic have been derived. In particular, it is proved in there that 
\eqref{mono} holds as soon as one has strong $L^1$ convergence of $\rho^\e$ and $J^\e$ in 
the limit $\e \to 0_+$. To show that this is indeed the case, we shall rely on the so-called modified WKB approximation introduced in \cite{Gr} and further developed in \cite{Ca}:
Define a complex-valued amplitude $a^\e$ by setting 
\begin{equation}\label{compla}
a^\e(t,x) = \psi^\e(t,x) e^{ -i S(t,x) / \e},
\end{equation}
where $\psi^\e$ solves \eqref{sch} and $S$ is a smooth solution of the Hamilton-Jacobi equation \eqref{hj}. 
Next, we recall that the results of \cite{Ca} (see also \cite{Ca1}) ensure that 
under our assumptions there is a time $T^*>0$, independent of $x\in \R^d$, such that, for all compact subsets  $I_t\subset [0, T^*)$,
the Hamiltonian flow \eqref{classflow} is well-defined, and
there exists a unique (sub-quadratic) phase function $S \in C^\infty (I_t\times \R^d)$,
given by \eqref{phi}. Consequently, this also ensures the 
existence of a smooth amplitude $a\in C^\infty(I_t; \mathcal S(\R^d))$ given by \eqref{a}.

With this result in hand, a straightforward computation shows that $a^\e$, defined in \eqref{compla}, solves 
\be\label{modamp}
\partial_t a^\e + \nabla a^\e  \cdot \nabla S + \frac{a^\e}{2} \Delta S= i \frac{\e}{2} \Delta a^\e,\quad a^\e(0,x) = a_0(x).
\ee
This equation can be considered as a perturbation of \eqref{wkbamp}. Indeed, if we denote the difference by $w^\e := a^\e - a$, then $w^\e$ satisfies
\[
\partial_t w^\e + \nabla w^\e  \cdot \nabla S + \frac{w^\e}{2} \Delta S= i \frac{\e}{2} \Delta a^\e,\quad w^\e(0,x) = 0,
\]
where the source term on the right hand side is formally of order $\mathcal O(\e)$.
Invoking energy estimates, one can prove (see \cite[Proposition 3.1]{Ca})
that for any time-interval $I_t\subset [0,T^*)$, there exists a unique solution
$a^\e \in C(I_t; H^s(\R^d))$ of \eqref{modamp},
and that
\[
\| w^\e \|_{L^\infty(I_t; H^s(\R^d))}\equiv \| a^\e - a \|_{L^\infty(I_t; H^s(\R^d))} = \mathcal O(\e), \quad \forall \,s\ge0.
\]
Writing the mass and current densities as 
\[
\rho^\e = |\psi^\e|^2= |a^\e|^2, \quad J^\e = \e \im\big(\overline{\psi^\e}\nabla \psi^\e \big) = |a^\e|^2 \nabla S+ \e \im\big(\overline{a^\e}\nabla a^\e \big),
\]
and using the fact that $H^s(\R^d) \hookrightarrow L^\infty(\R^d)$ for $s> d/2$, this consequently implies
\[
\rho^\e  \stackrel{\e\rightarrow 0_+ }{\longrightarrow} \rho, \quad \text{in $L^\infty(I_t;L^1(\R^d))$ strongly,}
\]
and 
\[
J^\e  \stackrel{\e\rightarrow 0_+ }{\longrightarrow} \rho u, \quad \text{in $L^\infty(I_t;L_{\rm loc}^1(\R^d)^d)$ strongly,} 
\]
where $\rho=|a|^2$ and $u=\nabla S$ are smooth solutions of the WKB system:
\begin{align*}
\partial_t \rho + \diver_x (\rho u) = 0, \quad \quad \rho(0,x) = |a_0(x)|^2,\\
\partial_t u + u\cdot \nabla u + \nabla V(x) = 0, \quad \quad u(0,x) = \nabla S_0(x).
\end{align*}
In particular, we infer that $P(t,y) = \nabla S(t,X(t,y))= u(t,X(t,y))$ 
and, in view of \eqref{a}, we also have that the density $\rho=|a|^2$ is given by 
\begin{equation}\label{rho}
\rho(t,x) = \frac{\rho_0(Y(t,x))}{J_t(Y(t,x))}\, , \quad t\in [0,T^*).
\end{equation}
The strong convergence of $\rho^\e, J^\e$ together with \cite[Theorem 3.6]{MPS1} then directly imply that the limiting measure 
$\beta$ is given by \eqref{mono} and thus Assertion (i) is proved.

In order to prove (ii) we first note that for every fixed $t \in [0,T^*)$, the limiting measure $\beta (t)$ is carried by 
the set $$\mathscr G_t= \{(x,p)\in \R^{2d} :  p= u(t,x)\}.$$The identity \eqref{formula} then implies that a.e. in $y$ the measure $\Upsilon_{t,y}$ is also carried by the same set and 
we consequently infer
\[
\Upsilon_{t,y}(x,p) = \mu_{t,y}(x) \delta (p - u(t,x)),
\]
where $\mu_{t,y}$ is the Young measure associated to $X^\e(t,y)$.

By taking the zeroth moment of $\Upsilon_{t,y}$ with respect to $p\in \R^d$ we realize that indeed 
$\mu_{t,y} = \upsilon_{t,y}$, with $\upsilon_{t,y}$ defined in Proposition \ref{proptrans}. We thus find that $\mu_{t,y}$ solves, in the sense of distributions:
\[
\partial_ t \mu_{t,y} + \diver_x \left(u \, \mu_{t,y} \right) = 0, \quad \mu_{t=0,y}(x) = \delta(x-y),
\]
a.e. with respect to the measure $\rho_0(y)$. In other words, $\mu_{t,y}(x)$ solves the same transport equation as the limiting density $\rho(t,x)$ does. In view of \eqref{rho}, we therefore 
conclude that, before caustic onset, $\mu_{t,y}$ is given by
\[
\mu_{t,y} (x)= \frac{1}{J_t(Y(t,x))} \delta(Y(t,x) - y),\quad \rho_0 - a.e..
\]
Multiplying by a test function $\varphi \in C_0(\R^d_x\times \R^d_y)$ and performing the change of variable $x = Y(t,x)$, we consequently find
\[
\langle \mu_{t,y} , \varphi \rangle = \iint_{\R^{2d}} \frac{1}{J_t(Y(t,x))} \, \delta(Y(t,x) - y) \varphi(x,y) \, dx \, dy = \int_{\R^d} \varphi(X(t,y), y) \, dy,
\]
and thus we can also express $\mu_{t,y}  = \delta(x-X(t,y))$. In summary we obtain that
\begin{equation*}\label{delta}
\Upsilon_{t,y}(x,p) = \delta(x- X(t,y)) \delta (p - u(t,X(t,y))).
\end{equation*}
a.e. on $\supp \, \rho_0 \subseteq\R^d$. In other words, the Young measure $\Upsilon_{t,y}$ is supported in a single point (on phase space). By a well known result in measure theory, cf. \cite[Proposition 1]{Hu}, this is 
equivalent to the local in-measure convergence of the associated family of trajectories $X^\e, P^\e$ and we are done. 
\end{proof}
The proof in particular shows, that, at least before caustic onset, the Young measure $\Upsilon_{t,y}$ is independent of the choice of $\rho_0^\e$, even 
though the Bohmian flow $X^\e$ is not.

\begin{remark} 
It is certainly possible to obtain Theorem \ref{th1} under weaker regularity assumption on $V, a_0^\e$, and $S_0$, which are imposed here only for the sake of simplicity. 
The assumption of $V$ and $S_0$ being sub-quadratic, however, can not be relaxed, if one wants to guarantee the existence of a non-zero caustic onset time 
$T^*>0$ {\it uniformly} in $x\in \R^d$, see, e.g., \cite{Ca} for a counter-example. 
Explicit examples of initial phases $S_0$, for which $T^*=+\infty$ (i.e., no caustic) are 
easily found in the case $V(x)\equiv 0$. Namely, either plane waves: $S_0 (x) = k\cdot x$, where $k\in \R^d$ is a given wave vector, or $S_0(x) = - |x|^2$, 
yielding a rarefaction wave for $t\in \R_+$, see \cite{GaMa}. In these situations, we obtain 
in-measure convergence of the Bohmian trajectories $(X^\e, P^\e)$, and consequently also uniform convergence of $X^\e$, 
locally on every Borel set $\Omega \subseteq \{ \R_t \times \supp \, \rho_0\}$ with finite Lebesgue measure.
\end{remark}

\section{Superposition of WKB states and Bohmian measures}\label{sec:super}

\subsection{Bohmian measure for multi-phase WKB states} 

In view of Fig.~\ref{caustic}, we expect that for $|t|> T^*$, i.e., after caustic onset, the correct asymptotic description of $\psi^\e$ is given by 
a superposition of WKB states, also known as {\it multi-phase ansatz}. In order to gain more insight in situations where this is indeed the case 
we shall, as a first step, study the classical limit of the corresponding Bohmian 
measure. To this end, let $\Omega \subset  \R_t\times \R^d_x$ be some open set and consider $\psi^\e$ to be given in the following form:
\begin{equation}
\label{eq:WKBaftercaustics}
\psi^\e(t,x)=\sum_{j=1}^N b_j(t,x) e^{iS_j(t,x)/\e} +r_\e(t,x),
\end{equation}
where $b_j\in C^\infty (\Omega;\C)$ are some smooth amplitudes and the real-valued phases
$S_j\in C^\infty( \Omega;\R)$ locally solve
\begin{equation}
\label{eq:HJ}
\partial_t S_j+\frac12 |\nabla S_j|^2+V(x)=0 \qquad \text{for all $j=1,\ldots,N$,}
\end{equation}
In addition, $r_\e$ denotes a possible remainder term (the assumptions on which will be made precise in the theorem below). 

\begin{remark} As we shall see Section \ref{sec:free},
the multi-phase WKB form \eqref{eq:WKBaftercaustics} can be rigorously established, locally on every
connected component of 
$(\mathbb R_t\times \R_x^d)\setminus \mathscr C$, i.e., locally away from caustics.
\end{remark}

The second main results of this work establishes an explicit formula for the limiting Bohmian
measure $\beta$ associated to a wave function of the form
\eqref{eq:WKBaftercaustics}. More precisely we prove the following:

\begin{theorem}
\label{thm:after}
Let $\psi^\e$ be as in \eqref{eq:WKBaftercaustics}, with $b_j\in C^\infty (\Omega;\C)$, $S_j\in C^\infty(\Omega;\R)$, for all $j=1,\ldots,N$, 
where $\Omega\subset [0,T]\times \R^d$ denotes some open set. 
Assume, in addition, 
\begin{equation}
\label{eq:diff gradient}
\nabla S_j \neq \nabla S_k\qquad \text{ for all $j\neq k\in \{1, \ldots, N\}$},
\end{equation}
and that the remainder $r_\e(t,x)$ satisfies
\begin{equation}
\label{eq:rest}
\|r_\e\|_{L^2_{\rm loc}(\Omega)} =o(1),\qquad \| \e \nabla r_\e\|_{L^2_{\rm loc}(\Omega)} =o(1)\qquad \text{as $\e \to 0_+$.}
\end{equation}
Then 
$$
\beta^\e  \stackrel{\e\rightarrow 0_+ }{\longrightarrow} \beta(t,x,p), \ \text{in $L^\infty([0,T];\mathcal M^+(\R_x^d\times \R^d_p)) \, {\rm weak-}\ast$,}
$$
where $\beta$ is given by
$$
\beta(t,x,p)= \int_{\T^{N}} \Gamma(t,x,\theta)
\, \delta\biggl(p-\frac{\sum_{j,k=1}^N \nabla S_j(t,x) \Gamma_{j,k}(t,x,\theta) }{\Gamma(t,x,\theta)}
\biggr)\,d\theta.
$$
with $\theta=(\theta_1,\ldots,\theta_N) \in \T^N$, and
\begin{equation}
\label{eq:def gamma}
\Gamma(t,x,\theta):=\biggl|\sum_{j=1}^N b_j(t,x)  e^{i\theta_j} \biggr|^2,
\qquad
\Gamma_{j,k}(t,x,\theta):=\re\left(b_j \bar b_k e^{i(\theta_{j} - \theta_{k})} \right).
\end{equation}
\end{theorem}
The above formula for $\beta$ generalizes equation (6.6) given in \cite{MPS1} and states that $\beta$ in general is a {\it diffuse measure} in the momentum variable $p \in \R^d$, unless 
all but one $b_j=0$. 
Note that, in the case where $N=1$, $\beta$ simplifies to a mono-kinetic phase space measure, i.e.,
$$
\beta(t,x,p) = |b(t,x)|^2 \delta(p-\nabla S(t,x)).
$$
We already know from Assertion (i) of Theorem \ref{th1} that this holds for $|t|<T^*$, i.e., before caustic onset.

\begin{proof} By our assumptions, it is easy to check that 
$\rho^\e=|\psi^\e|^2= \tilde \rho^\e+r_{1,\e}$ and $J^\e= \e \im\big(\overline{\psi^\e}(t,x)\nabla \psi^\e(t,x)\big)=
\tilde J^\e+r_{2,\e}$, where
$$
\tilde \rho^\e:= \sum_{j,k=1}^N b_j \bar b_k e^{i(S_j-S_k)/\e},\quad \tilde J^\e:=
\sum_{j,k=1}^N \nabla S_j \, \re\left(b_j \bar b_k e^{i(S_j-S_k)/\e}\right).
$$
and
$$
\|r_{1,\e}\|_{L^1_{\rm loc}(\Omega)} =o(1),
\qquad \|r_{2,\e}\|_{L^1_{\rm loc}(\Omega)} =o(1).
$$
In order to derive the classical limit as $\e \to 0_+$ of the Bohmian measure $\beta^\e$, we need to compute the limit of expressions of the following form
\begin{equation}
\label{eq:beta eps}
\iint_{\Omega} \sigma(t,x) \rho^\e(t,x) \varphi\Bigl(t,\frac{J^\e(t,x)}{\rho^\e(t,x)} \Bigr)\,dx\,dt,
\end{equation}
where $\varphi,\sigma \in C_{\rm c}^\infty([0,T]\times \R^d;  \R)$ are smooth and compactly supported. To this end, we first note that, because $\varphi$ is smooth and compactly supported, the map
$$
\R^+ \times \R^d \ni (s,v) \mapsto s \varphi \Bigl(t,\frac{v}{s} \Bigr)
$$
is Lipschitz (uniformly with respect to $t$), which implies
$$
\biggl\|\rho^\e \varphi \Bigl(t,\frac{J^\e}{\rho^\e} \Bigr)-
\tilde \rho^\e \varphi \Bigl(t,\frac{\tilde J^\e}{\tilde \rho^\e} \Bigr)\biggr\|_{L^1_{\rm loc}(\Omega)}
\le C \Bigl(\|r_{1,\e}\|_{L^1_{\rm loc}(\Omega)} + \|r_{2,\e}\|_{L^1_{\rm loc}(\Omega)} \Bigr)=o(1).
$$
In particular, to compute the limit as $\e \to 0_+$ of the expression in \eqref{eq:beta eps}
it suffices to consider
\begin{equation}
\label{eq:beta eps2}
\iint_{\Omega} \sigma(t,x) \tilde \rho^\e(t,x) \varphi \Bigl(t,\frac{\tilde J^\e(t,x)}{\tilde \rho^\e(t,x)} \Bigr)\,dx\,dt.
\end{equation}
We now use the following result, whose proof is postponed to the end.

\begin{lemma}
\label{lem:rational}
There exists a set $\Sigma\subset \Omega$ of $\mathscr L^{d+1}$-measure zero such that, for all $j,k,\ell \in \{1,\ldots,N\}$
with $k \neq \ell$,
$$
\frac{S_j(t,x)-S_k(t,x)}{S_j(t,x)-S_\ell(t,x)} \not \in \Q \qquad \text{for all
$(t,x) \in \Omega \setminus \Sigma$}.
$$ 
\end{lemma}
Using this lemma, we deduce that for $\mathscr L^{d+1}-a.e.$ $(t,x)$, the frequencies
$$\frac{S_1(t,x)-S_k(t,x)}{\e}, \qquad k = 2, \dots, N,$$ are
all rationally independent, which implies that the ``trajectories"
$$
\e \mapsto \left(\cos\Bigl(\frac{S_2-S_1}\e\Bigr), \ldots, \cos\Bigl(\frac{S_N-S_1}\e\Bigr)\right)
$$
and
$$
\e \mapsto \left(\sin\Bigl(\frac{S_2-S_1}\e\Bigr), \ldots, \sin\Bigl(\frac{S_N-S_1}\e\Bigr)\right)
$$
are both dense on the $(N-1)$-dimensional torus $\T^{N-1}$. By standard results
on two-scale convergence (see for instance \cite{Al}), we consequently obtain that for any continuous and compactly supported test function
$\vartheta: \Omega \times \C^{N-1} \to \R$,
\begin{multline*}
\int_{\Omega} \vartheta
\left(t,x,e^{i(S_2-S_1)/\e}, \ldots, e^{i(S_N-S_1)/\e}\right)\,dx\,dt\\
  \stackrel{\e\rightarrow 0_+ }{\longrightarrow} \int_{\Omega}\int_{\T^{N-1}} \vartheta\left(t,x,e^{i\theta_1},
\ldots,e^{i\theta_{N-1}}\right)
\,d\theta_1\ldots d\theta_{N-1}\,dx\,dt.
\end{multline*}
Moreover, we observe that for any $j,k$ we can write
$$
\frac{S_j-S_k}\e = \frac{S_j-S_1}\e+\frac{S_1-S_k}\e.
$$
Hence the expression in \eqref{eq:beta eps2} converges to
\begin{multline*}
\iint_{\Omega} \sigma(t,x) \int_{\T^{N-1}} \sum_{j,k=1}^N b_j \bar b_k e^{i(\theta_{j-1} - \theta_{k-1})}\\
\varphi\Biggl(t,\frac{\sum_{j,k=1}^N\nabla S_j\,
\re\left(b_j \bar b_k e^{i(\theta_{j-1} - \theta_{k-1})}\right)}{\sum_{j,k=1}^N b_j \bar b_k e^{i(\theta_{j-1} - \theta_{k-1})}}
\Biggr)\,d\theta_1\ldots d\theta_{N-1}\,dx\,dt,
\end{multline*}
where by convention $\theta_0\equiv 0$. 
Finally, let us observe that one can also rewrite the obtained expression in a more symmetric form by performing the change of variables
$\theta_{j-1} \leftrightarrow \theta_j -\theta_1$, and 
it is immediate to check that under this transformation the above expression is equal to
\begin{align*}
\iint_{\Omega} \sigma(t,x) \int_{\T^{N}} \Gamma(t,x,\theta)\,
\varphi\biggl(t,\frac{\sum_{j,k=1}^N\nabla S_j \Gamma_{j,k}(t,x,\theta)}{\Gamma(t,x,\theta)}
\biggr)\,d\theta\,dx\,dt,
\end{align*}
where $\theta=(\theta_1,\ldots,\theta_N)$, and $\Gamma$ and $\Gamma_{j,k}$ are defined in \eqref{eq:def gamma}.
By the arbitrariness of $\varphi$ and $\sigma$, this proves the desired result.
\end{proof}

We are now left with the proof of Lemma \ref{lem:rational}.
\begin{proof}[Proof of Lemma \ref{lem:rational}]
The set $\Sigma$ can be described as 
$$
\bigcup_{j,k,\ell,\, k \neq \ell} \bigcup_{m\neq n \in \Z} S_{j,k,\ell}^{m,n},
$$
where
$$
S_{j,k,\ell}^{m,n}:=\bigl\{(t,x) \in \Omega\,:\,m[S_j(t,x)-S_k(t,x)]+n[S_j(t,x)-S_\ell(t,x)]=0\bigr\}.
$$
We now claim that each $S_{j,k,\ell}^{m,n}$ is a smooth hypersurface in $\Omega$, which
implies in particular that $S_{j,k,\ell}^{m,n}$ (and so also $\Sigma$) has measure zero. 
To prove that this is indeed the case, it suffices to check, in view of the implicit function theorem, that the gradient of the function
$$
(t,x)\mapsto m[S_j(t,x)-S_k(t,x)]+n[S_j(t,x)-S_\ell(t,x)]
$$
is nowhere zero. Assume by contradiction that this is not the case, i.e., there exists a point $(t,x) \in \Omega$ where
$$
(m+n)\partial_t S_j(t,x)=m\partial_tS_k(t,x)+n\partial_tS_\ell(t,x),
$$
$$
(m+n)\nabla S_j(t,x)=m\nabla S_k(t,x)+n\nabla S_\ell(t,x).
$$
By \eqref{eq:HJ}, the first equation above becomes
$$
(m+n)|\nabla S_j(t,x)|^2=m|\nabla S_k(t,x)|^2+n|\nabla S_\ell(t,x)|^2,
$$
which combined with the second equation gives
$$
(m+n)\left|\frac{m}{m+n}\nabla S_k(t,x)+\frac{n}{m+n}\nabla S_\ell(t,x)\right|^2=
m|\nabla S_k(t,x)|^2+n|\nabla S_\ell(t,x)|^2.
$$
By strict convexity of $|\cdot|^2$, the above relation is possible if and only if $\nabla S_k(t,x)=\nabla S_\ell(t,x)$,
which contradicts \eqref{eq:diff gradient} and concludes the proof.
\end{proof}

\subsection{Comparison to Wigner measures} \label{sec:wig}

An important consequence of Theorem \ref{thm:after} concerns the connection between the limiting Bohmian measure $\beta$ and the 
Wigner measure $w\in \mathcal M^+(\R^d_x\times \R^d_p)$ associated to $\psi^\e$. To this end, let us first recall the definition of the 
$\e$-scaled Wigner transform $w^\e$ given in \cite{ALMS,GMMP, LiPa}:
\begin{equation*}
w^\e (t,x,p): = \frac{1}{(2\pi)^d} \int_{\R^d}
\psi^\e\left(t,x-\frac{\e}{2}\eta
\right)\overline{\psi^\e} \left(t,x+\frac{\e}{2}\eta
\right)e^{i \eta \cdot p}\,  d\eta. 
\end{equation*}
Provided $\psi^\e(t)$ is uniformly bounded in $L^2$ with respect to $\e$, it is well known that, cf. \cite{GMMP, LiPa} there exists a limit $w(t,x,p)$ such that
$$
w^\e  \stackrel{\e\rightarrow 0_+ }{\longrightarrow} w, \ \text{in $L^\infty(\R_t;\mathcal M^+(\R_x^d\times \R^d_p)) \, {\rm weak-}\ast$.}
$$
In addition, one finds $w(t) \in \mathcal M^+(\R^d_x \times \R^d_p)$, usually called \emph{Wigner measure} (or semi-classical defect measure). 
The latter is known to give the possibility to compute the classical limit of all expectation values of physical observables via
\begin{equation*}
\lim_{\e\to 0} \langle \psi^\e(t), \text{Op}^\e(a)\psi^\e(t)\rangle_{L^2(\R^d)} =  \iint_{\R_{x,p}^{2d}} a(x,p) w(t, x,p) \, dx \, dp,
\end{equation*}
where the $ \text{Op}^\e(a)$ is a self-adjoint operator obtained by Weyl-quantization of the corresponding classical symbol 
$a\in \mathcal S(\R^d_x \times \R^d_p)$, see \cite{GMMP, SMM} for a precise definition. 
In addition, if $\psi^\e(t)$ is $\e$-oscillatory, i.e., satisfies \eqref{epsosc}, we also have that the zeroth and first $p$-moment of $w$ yield the classical limit of 
$\rho^\e$ and $J^\e$, i.e.,
\[
\rho^\e(t,x) \stackrel{\e\rightarrow 0_+ }{\longrightarrow} \int_{\R^d}w(t, x, dp), \quad J^\e(t,x) \stackrel{\e\rightarrow 0_+ }{\longrightarrow} \int_{\R^d} p w(t, x, dp),
\]
where the limits have to be understood in $L^\infty(\R_t; \mathcal M^+(\R^d_x)) \, {\rm weak-}\ast$. Note that this is indeed {\it analogous} to \eqref{Bmoment}.

For a given superposition of WKB states such as \eqref{eq:WKBaftercaustics},
the associated Wigner measure has been computed in \cite{LiPa} (see also \cite{SMM}):
under the same assumption on the phases, i.e., $\nabla S_j \neq \nabla S_k$ for all $j\neq k$, one explicitly finds
\begin{equation}\label{wigdelta}
w(t,x,p) = \sum_{j=1}^N |b_j(t,x)|^2 \delta(p-\nabla S_j(t,x)).
\end{equation}
>>From this explicit formula we immediately conclude the following important corollary.

\begin{corollary} \label{cor:wig}
Let $b_j \not =0$. Then, under the same assumptions as in Theorem \ref{thm:after} we have that, in the sense of measures, $\beta = w$ if and only if $N=1$.
\end{corollary}

\begin{proof} For $b_j\not =0$ and $N>1$ we see from Theorem \ref{thm:after} that $\beta$ is a diffuse measure in the momentum variable $p\in \R^d$, and thus $\beta \not = w$ in view of \eqref{wigdelta}. 
On the other hand, if $N=1$ then, both $w$ and $\beta$ simplify to the same mono-kinetic phase space distribution.
\end{proof}

In view of Assertion (i) of Theorem \ref{th1} we conclude that before caustic onset, the classical limit of all physical observables can be computed by taking 
moments of the limiting Bohmian measure, since in fact $\beta = w$ for $|t|<T^*$. After caustic onset, however, this is in general \emph{no longer} the case (see Section \ref{sec:free}). 

Still, we do know (by weak compactness arguments) that the {\it zeroth} and {\it first moments}
w.r.t. $p\in \R^d$ of $\beta$ and $w$ are {\it the same} for all times $t\in \R$. 
For completeness, we check this explicitly in the case of multi-phase WKB states: using the fact that
$$
\int_{\T}\cos\left(\theta\right)\,d\theta=\int_{\T}\sin\left(\theta\right)\,d\theta=0 ,
$$
we compute
\begin{align*}
 \int_{\R^d} \beta(t,x,dp)&=\sum_{j, k=1}^N
b_j  (t,x) \bar b_k (t,x) \int_{\T^N} e^{i(\theta_{j}-\theta_{k})}\,d\theta_1\ldots d\theta_N\\
&=\sum_{j=1}^N |b_j(t,x)|^2
=\int_{\R^d} w(t,x,dp).
\end{align*}
Moreover
\begin{align*}
 \int_{\R^d} p \beta(t,x,dp)&=\sum_{j, k=1}^N
\nabla S_j(t,x) \int_{\T^N} \re\left( b_j  (t,x) \bar b_k (t,x)  e^{i(\theta_{j}-\theta_{k})}\right)\,d\theta_1\ldots d\theta_N\\
&=\sum_{j=1}^N \nabla S_j(t,x) |b_j(t,x)|^2
=\int_{\R^d} p w(t,x,dp).
\end{align*}
In other words, in the case of multi-phase
WKB states, the difference between $w$ and $\beta$ can only manifest itself in $p$-moments of order two or higher.

\section{A complete description in the free case and possible extensions}\label{sec:free}

In this section we shall give a (fairly) complete description of the classical limit of Bohmian dynamics in the case of the {\it free Schr\"odinger equation} corresponding to $V=0$. The proof will  
rely on classical stationary phase techniques. For the case $V\not=0$ 
decisively more complicated methods based on Fourier integral operators have to be employed, as will be discussed in Section \ref{sec:FIO}.

\subsection{Multi-phase WKB for vanishing potential} Consider the free Schr\"odinger equation with WKB initial data:
\begin{equation}
\label{schfree}
i \e \partial_t  \psi^\e   + \frac{\e^2}{2}\Delta \psi^\e =0,\quad \quad 
\psi^\e|_{t=0}   = a_{0}(x) e^{i S_0(x)/\e},
\end{equation} 
In this case, we find the free Hamilton-Jacobi equation, which is obviously given by
\begin{equation}\label{freehj}
\partial S + \frac{1}{2} |\nabla S|^2 = 0, \quad S|_{t=0} = S_0,
\end{equation}
and the corresponding classical Hamiltonian equations \eqref{classflow} simplify to
\begin{equation}\label{freeflow}
\left \{
\begin{aligned}
& \,  \dot X(t,y)= P(t,y) , \quad X(0,y)=y,\\
& \, \dot P(t,y)= 0 , \quad P(0,y)=\nabla S_0(y).
\end{aligned}
\right. 
\end{equation}
This implies that, for all $t\in \R_+$, $P(t,y) = \nabla S_0(y)$ and 
\begin{equation}\label{freeX}
 X(t,y) = y + t \nabla S_0(y) .
\end{equation}
Consequently, the caustic set is given by $ \mathscr  C_{\rm free}= \{(x,t) : x\in  \mathscr  C^{\rm free}_t \}$ where for $x\equiv X(t,y)$ we set:
\[
 \mathscr{ C}^{\rm free}_t =
\bigl\{x\in \R^d \, : \, \exists \text{ $y\in \R^d$ satisfying \eqref{freeX} and $\det(  {\rm Id} + t \nabla^2 S_0(y)) = 0$} \bigr\}.
\]
In particular, we see that in the free case, the caustic onset time $T^*>0$ is solely determined by the (sub-quadratic) initial phase $S_0(y)$. 
In order to proceed we need to slightly strengthen our assumption on the initial phase $S_0$.
\begin{hyp}\label{hypSfree} 
The initial phase $S_0\in C^\infty(\R^d;\R)$ is assumed to be sub-quadratic and
$$
\lim_{|y| \to \infty} \frac{|\nabla S_0(y)|}{|y|} = 0.
$$
\end{hyp}
In other words we need that $S_0$ grows strictly less than quadratically at infinity. This is the same assumption as in \cite{BGMP}, guaranteeing that the map $y\mapsto X(t,y)$ is proper and onto.

In the following we shall denote by $x\mapsto y\equiv Y(t,x)$ the inverse mapping of \eqref{freeX}. Clearly, for $|t|>T^*$ this inverse will not be unique in general,
i.e., 
for each fixed $(t,x)\in \R_t\times \R^d_x$ there is $N(t,x)\in \N$ and corresponding $Y_j(t,x)$, with $j=1, \dots, N(t,x)$, satisfying the implicit relation
\begin{equation}\label{freeY}
Y_j(t,x) + t \nabla S_0(Y_j(t,x)) = x.
\end{equation}
Assumption \ref{hypSfree} guarantees that in each connected component of $(\R_t\times \R^d) \setminus \mathscr C_{\rm free} $ there are only
finitely many $\{Y_j(t,x)\}$. (This 
follows by properness of the characteristic map and the implicit function theorem, see \cite[Lemma 1.1]{BGMP}.) In addition,
in each such connected component $N(t,x)= \text{const}$.
Moreover, under the same assumptions on $S_0$, we already know that the caustic onset time $T^*$ is positive,
and thus
there is exactly one connected component $\Omega_0$ of $ (\R_t\times \R^d) \setminus \mathscr C_{\rm free} $ containing $\{ t = 0 \}$.

In order to proceed further, we also recall that the solution of \eqref{schfree} admits an explicit representation in the form of an $\e$-oscillatory integral
\begin{equation}\label{explicit}
\psi^\e(t,x) = \left(\frac{1}{\sqrt{2\pi i \e t}}\right)^d \int_{\R^d} a_0(y) e^{i \Phi(t,x,y)/\e} \, dy,
\end{equation}
where the phase is given by
\begin{equation}\label{Phi}
\Phi(t,x,y) := S_0(y) + \frac{|x-y|^2}{2t}.
\end{equation}
It is well known that, for $\e\to 0_+$, the representation formula \eqref{explicit} can be treated by the stationary phase techniques (see, e.g., Theorem 7.7.6. of \cite{Ho}) and we consequently obtain the following lemma.

\begin{lemma} \label{lem:statphase}
Let $a_0\in \mathcal S(\R^d;\C)$ and $S_0 $ satisfy Assumption \ref{hypSfree}. Then, for all $(t,x) \in  (\R_t\times \R^d) \setminus \mathscr C_{\rm free} $ the solution of 
\eqref{schfree} satisfies
\begin{equation}\label{statpsi}
\psi^\e(t,x) \stackrel{\e \to 0_+}{= }\sum_{j=1}^{N(t,x)} a_j(t,x) e^{i \pi \kappa_j (t,x)/4} \, e^{i \Phi(t,x, Y_j(t,x))/\e} +r_\e(t,x),
\end{equation}
where $\Phi(t,x,y)$ is given by \eqref{Phi}, $\kappa_j(t,x)\in \N$ denotes the Maslov factor, and 
\begin{equation}\label{aj}
a_j(t,x) =  \frac{a_0(Y_j(t,x))}{|\det( {\rm Id} + t \nabla^2 S_0(Y_j(t,x)))|^{1/2}}.
\end{equation}
In addition, the remainder $r^\e$ satisfies 
\begin{equation}
\label{eq:rest1}
\|r_\e\|_{C^0(\Omega)} =\mathcal O(\e),\qquad \| r_\e\|_{C^1(\Omega)} =\mathcal O(1)\qquad \text{as $\e \to 0_+$,}
\end{equation}
uniformly on compact subsets $\Omega \subset (\R_t\times \R^d) \setminus \mathscr C_{\rm free} $.
\end{lemma}
\begin{remark} The first remainder estimate $\|r_\e\|_{C^0(\Omega)} =\mathcal O(\e)$ is classical, whereas the second one 
can be obtained by noticing that the operator $\nabla$ commutes with the free Schr\"odinger equation \eqref{schfree}. 
Thus, we find that $\nabla \psi^\e$ satisfies an integral representation analogous to \eqref{explicit}, i.e.,
\[
\nabla \psi^\e(t,x) = \left(\frac{1}{\sqrt{2\pi i \e t}}\right)^d \int_{\R^d}  e^{i |x-y|^2/(2t \e)} \, \nabla \psi^\e_0(y) \, dy.
\]
By applying the stationary phase lemma to this oscillatory integral one readily infers the estimate $\| r_\e\|_{C^1(\Omega)} =\mathcal O(1)$.
\end{remark}

Next, we note that, in view of \eqref{Phi} and \eqref{freeY}, we explicitly have
\begin{equation}
\begin{split}\label{PhiJ}
 \Phi(t,x, Y_j(t,x)) \equiv& \ S_0(Y_j(t,x)) + \frac{1}{2t} | x - Y_j(t,x)|^2  \\= & \ S_0(Y_j(t,x)) + \frac{t}{2} |\nabla S_0(Y_j(t,x)|^2.
\end{split}
\end{equation}
On the other hand, since for $V(x)= 0$ it holds that $P(t,y) = \nabla S_0(y)$ (i.e., $P$
is constant along the characteristics),
the solution formula \eqref{phi} yields, for all $j=1, \dots, N$:
\begin{equation}
\begin{split}
S_j(t,x) & \ = S_0(Y_j(t,x)) + \int_0^t \frac{1}{2} |P(\tau, y)|^2  \  d \tau \big|_{y=Y_j(t,x)} \\
& \ = S_0(Y_j(t,x)) + \frac{t}{2} |\nabla S_0(Y_j(t,x)|^2.
\end{split}
\end{equation}
We consequently infer that $ \Phi(t,x, Y_j(t,x)) \equiv S_j(t,x)$ is a smooth solution of the free Hamilton-Jacobi equation \eqref{freehj} for all $j =1, \dots , N(t,x)$. 
Obviously, we also have that $a_j$ given by \eqref{aj} solves the corresponding transport equation \eqref{a} with $S\equiv S_j$. 

\begin{remark} An alternative way of showing that $ \Phi(t,x, Y_j(t,x))$ solves the free Hamilton-Jacobi equation 
is to plug \eqref{PhiJ} into \eqref{freehj} and use \eqref{freeY} to implicitly differentiate with respect to $t$ and $x$. A lengthy but straightforward computation then yields the desired result.
\end{remark}
For completeness we also recall that the Maslov factor is explicitly given by \cite{Ho}
\[
\N \ni \kappa_j(t,x) = m_j^+ (t,x) - m_j^-(t,x) ,
\]
where $m^\pm(t,x)\in \N$ denotes, respectively, the number of positive or negative eigenvalues of the matrix
${\rm Id} + t \nabla^2 S_0(Y_j(t,x))$. 
Note that $\kappa_j$ can also be written in the form 
\[
\kappa_j(t,x) =  d - 2 m_j^-(t ,x).
\]
By the implicit function theorem, $\kappa(t,x)= \text{const}$ in every connected component
of $(\R_t \times \R^d_x) \setminus \mathscr C_{\rm free}$, see, e.g., \cite{BGMP}.

\subsection{WKB analysis of Bohmian dynamics in the free case}\label{sec:freebohm} From what is said above, we infer that 
in each connected component $\Omega$ of $(\R_t \times \R^d_x) \setminus \mathscr C_{\rm free}$, the solution $\psi^\e$ 
admits the approximation \eqref{statpsi}, so
Theorem \ref{thm:after} can be applied after identifying
$$b_j (t,x)= a_j (t,x)e^{i \pi \kappa_j (t,x)/4}\equiv a_j (t,x)e^{i \pi \kappa_\Omega/4},
\qquad j = 1, \dots, N(t,x)\equiv N_\Omega,$$ 
where $\kappa_\Omega \in \R$ and $N_\Omega \in \N$ are constants depending only on $\Omega$.
Consequently, we obtain the following result.

\begin{theorem} \label{thm:free}
Let $a_0\in \mathcal S(\R^d;\C)$ and $S_0 $ satisfy Assumption \ref{hypSfree}. Denote by $\Omega_0 $ the 
connected component of $ (\R_t\times \R_x^d) \setminus \mathscr C_{\rm free} $ containing $\{ t = 0 \}$. Then it holds:

\emph{(i)} The limiting Bohmian measure satisfies
\[
\beta (t,x,p) = w(t,x,p) = \rho(t,x) \delta (p - u(t,x)), \quad \forall (t,x) \in \Omega_0 ,
\]
and the Bohmian trajectories converge
\[
X^\e (t,y)  \stackrel{\e\rightarrow 0_+ }{\longrightarrow}  y + t \nabla S_0(y), \quad P^\e (t,y)  \stackrel{\e\rightarrow 0_+ }{\longrightarrow}  \nabla S_0(y), 
\]
locally in measure on $\Omega_0 \cap \{ \R_t \times \supp \rho_0\}$.

\emph{(ii)} Outside of $\Omega_0$ there are regions $\Omega\subseteq (\R_t\times \R_x^d) \setminus \mathscr C_{\rm free} $ where $\beta \not = w$ and where the 
Bohmian momentum $P^\e$ does not converge locally in-measure to the classical momentum $P$.

\emph{(iii)} There exist initial data $a_0(y)$ and $S_0(y)$ such that, outside of $\Omega_0$, there are regions $\tilde \Omega\subseteq (\R_t\times \R_x^d) \setminus \mathscr C_{\rm free} $ in which 
both $X^\e$ and $P^\e=\dot X^\e$ do not converge to the classical flow. \end{theorem}

Note that Assertion (i) is slightly stronger than Theorem \ref{th1} (i) in the sense that $\Omega_0$ is strictly {\it larger} than $[0,T^*) \times \R^d_x$. The proof shows that if $|a_0|>0$ on all of $ \R^d$, 
Assertion (ii) holds for any connected component $\Omega \not = \Omega_0$, whose boundary intersects the boundary of $\Omega_0$.

\begin{proof} 
We first note that for all $(t,x)\in \Omega_0$ it holds $N(t,x) = 1$ and $\kappa_j (t,x) =0$. 
In view of the remainder estimates stated in Lemma \ref{lem:statphase} we thus can apply Theorem \ref{thm:after} with $N=1$ to obtain
$$\beta (t,x,p)= \rho(t,x) \delta (p -\nabla S(t,x)),$$
where $\rho =  |a|^2 $. 
With this in mind, the result on the convergence of the Bohmian trajectoriess follows verbatim from
the proof of Theorem \ref{th1} (ii).  This proves the first assertion.

In order to prove Assertion (ii), we first note that that outside of $\Omega_0$ we have (in general) more
than one branch, i.e., $N(t,x) > 1$. For instance, assume that $|a_0|>0$ on $\R^d$,
and let $\Omega\neq \Omega_0$ be a connected component whose boundary
intersects the boundary of $\Omega_0$. Then it is not difficult to see that $N_\Omega \neq 1$,
as otherwise one could show that no caustics can occur on $\partial \Omega_0 \cap \partial \Omega$.
Next, we recall that in each connected component $ \Omega$ of $ (\R_t\times \R^d) \setminus \mathscr C_{\rm free} $ 
the phase $ \Phi(t,x, Y_j(t,x)) \equiv S_j(t,x)$ is a smooth solution of the Hamilton-Jacobi equation \eqref{hj}.
By the method of characteristics we have that
$$ \nabla \Phi(t,x,Y_j(t,x)) \equiv \nabla S_j(t,x) = P(t,Y_j(t,x)) = \nabla S_0(Y_j(t,x)),$$
since $P(t,y)$ is constant along characteristics (recall that $V(x)=0$). 
Hence, assuming by contradiction that $\nabla S_j = \nabla S_k$ for some $j\not = k$ , the above identity
together with \eqref{freeX} yields $Y_j (t,x)= Y_k(t,x)$, which is impossible by construction. 
This implies that in each connected component $\Omega$ we can apply Theorem \ref{thm:after} to conclude that
$\beta$ in general 
is a diffuse measure in $p\in \R^d$, unless all but one of the $a_j=0$ in $\Omega$. In view of \eqref{aj},
the latter cannot be the case if $|a_0|>0$ on $\R^d$.
Corollary \ref{cor:wig} then immediately implies $\beta \not = w$. 
On the other hand, since for WKB initial data 
we have that $\rho^\e_0$ is indeed $\e$-independent, we can apply \eqref{formula} in $\Omega$ to infer that the Young measure $\Upsilon_{t,y}$ is diffusive in $p$ (since $\beta$ is). 
This, however, prohibits the convergence of $P^\e$ locally in measure, since the latter is equivalent to $\Upsilon_{t,y}$ being concentrated in a single point.

The result in (ii) may still give some hopes for the convergence of $X^\e$ to $X$, since the fact that
$\dot X^\e=P^\e$ gives more compactness for the curves in the $x$-variables. However,
we shall see that this is not the case.

Consider indeed the example described in Fig. \ref{caustic} and Fig. \ref{Xcau1e3} (so $d=1$).
These figures suggest that for $\psi^\e_0$ as in \eqref{numfreeini} convergence should not hold.
To show this rigorously, we begin by observing that
$\rho(t,x)>0$ on $\R_t \times \R_x$ (this follows from the explicit formula
for $\rho=|a|^2$, but it can also be seen from Fig. \ref{caustic} observing there
only the trajectories starting inside $[0,1]$ are plotted).
Since $\rho$ is smooth, this implies that for $R,T>0$ there exists a positive constant $c_{R,T}$ such that
$$
\rho(t,x) \ge c_{R,T}\quad \text{for }(t,x)\in [0,T] \times [-R,R].
$$
In particular, since $\psi^\e$ is given by \eqref{statpsi} with $r_\e$ small in $C^0$, see \eqref{eq:rest1},
it follows that
\begin{equation}
\label{eq:bound rho e}
\rho^\e(t,x) \ge \frac{c_{R,T}}{2}\quad \text{for }(t,x)\in [0,T] \times [-R,R]
\end{equation}
for all $\e>0$ sufficiently small (the smallness depending on $T$ and $R$).
Recalling that
$$
\dot X^\e=u^\e(t,X^\e(t,x)),\quad u^\e=\frac{J^\e}{\rho^\e},
$$
and that $J^\e$ and $\rho^\e$ are both smooth, it follows from \eqref{eq:bound rho e}
that $u^\e$ is smooth as well inside $[0,T] \times [-R,R]$.
In particular, by the Cauchy-Lipschitz theorem, the Bohmian trajectories $X^\e$
can never cross inside $[0,T] \times [-R,R]$.
Since by symmetry
$X^\e(t,1/2)=1/2$ for all $t \ge 0$,
this implies in particular that, for all $t \in [0,T]$:
$$
X^\e(t,x) \ge 1/2 \quad \forall \,x \ge 1/2,\qquad X^\e(t,x) \le 1/2 \quad \forall \,x \le 1/2.
$$
Letting $\e \to 0$ we deduce that $X^\e \not \rightarrow X$ (locally) in measure on $\tilde \Omega \equiv [0,T] \times [-R,R]$, since otherwise
the above property would give
$$
X(t,x) \ge 1/2 \quad \forall \,x \ge 1/2,\qquad X(t,x) \le 1/2 \quad \forall \,x \le 1/2
$$
for all $t \ge 0$, which is not the case (see Fig. \ref{caustic}).
This proves Assertion (iii).
\end{proof}

\begin{remark} Note that for $|t|>T^*$, i.e., after caustic onset, the Wigner measure is given by \eqref{wigdelta} for all $(t,x)\in (\R_t \times \R^d_x) \setminus \mathscr C_{\rm free}$. 
In particular, this shows that $w$ is {\it insensitive} to the Maslov phase shifts, since $|a_j|^2=|b_j|^2$ for all $j=1, \dots, N(t,x)$. 
The limiting Bohmian measure $\beta$, however, {\it incorporates} these phase shifts in view of the formula given in Theorem \ref{thm:after}. 
However, as we have seen in Section \ref{sec:wig} these phase shift do not enter in the classical limit of $\rho^\e$ and $J^\e$.
\end{remark}

\subsection{Extension to the non-zero potential case}\label{sec:FIO}

In the case where $V(x)\not = 0$ the situation becomes considerably more complicated, due to a lack of an explicit integral representation for the exact solution $\psi^\e$ of \eqref{sch}. The only exception 
therefrom is the case $V(x) = \pm \frac{1}{2} |x|^2$ where one has {\it Mehler's formula} replacing \eqref{explicit}, see, e.g., \cite{Ca1}. In order to proceed further in situations where 
$V$ is a more general (sub-quadratic) potential, 
one needs to approximate the full Schr\"odinger propagator $$U^\e(t) = e^{-i H^\e t}, \quad \text{with}\ H^\e = -\frac{\e^2}{2} \Delta + V(x),$$ 
for $0<\e\ll 1$ by a semi-classical {\it Fourier integral operator} \cite{Du}. 
Early results on this can be found in \cite{Ch, Fu}, where the occurrence of caustics makes the approximation valid only locally in-time. 
This problem can be overcome, by considering a class of Fourier integral operators whose Schwartz kernel furnishes 
an $\e$-oscillatory integral with {\it complex} phase and quadratic 
imaginary part, see \cite[Theorem 2.1]{LS} for a precise definition. 
Using this, the authors of \cite{LS} construct a global in-time approximation of $U^\e(t)$ for potentials satisfying $V\in C^\infty_{\rm b}(\R^d)$,
i.e.,
smooth and bounded together with all derivatives (see also \cite{GL,KS} for closely related results with slightly different assumptions). 
By applying the stationary phase lemma to this type of (global) Fourier integral operator, one infers the following result, as a slight generalization of \cite[Theorem 5.1]{LS}:

Fix a point $(t_0,x_0)\in (\R_t \times \R^d_x) \setminus \mathscr C$, i.e., 
away from caustics, and as before denote by $Y_j(t,x)$ and $j=1, \dots, N=N(t,x)\in \N$,
the solutions of the equation $x= X(t,y)$, where $t \mapsto X(t,y)$ is 
the classical flow map induced by \eqref{classflow}. Let $\{ y \in \R^d  :  |a_0(y)| >0 \},$
be a sufficiently small neighborhood of 
\begin{equation}\label{ypsilons}
\{ Y_1(t_0, x_0), \dots, Y_N(t_0, x_0)\} \subset \R^d,
\end{equation}
i.e., the points obtained by tracing back the classical trajectories intersecting in
$(t_0,x_0)\in (\R_t \times \R^d_x) \setminus \mathscr C$.
Then the solution of \eqref{sch} at $t=t_0$ admits the following approximative behavior:
\begin{align}\label{multiWKB}
\psi^\e(t_0, x)  \stackrel{\e \to 0_+}{= }\sum_{j=1}^{N(t,x)} a_j(t_0,x) e^{i \pi (m^+_j(t_0,x) - m^-_j(t_0,x) )/4 } \, e^{i S_j(t_0,x)/\e} +r_\e(t_0,x),
\end{align}
where the amplitudes $a_j$ and the (real-valued) phases $S_j$ are, respectively, given by \eqref{a} and \eqref{phi} with $Y$ replaced by $Y_j(t_0,x)$, and 
$m^+_j(t_0,x)$ (resp. $m^-_j(t_0,x)$) is the number of positive (resp. negative) eigenvalues of the matrix $\nabla_y X_t(Y_j(t_0,x))$.
In addition, the remainder $r_\e$ satisfies
\[
\| r_\e(t_0,\cdot) \|_{L^2(\Lambda)} = \mathcal O(\e),
\]
where $x \in \Lambda \subset \R^d$ is a {\it sufficiently small} neighborhood of $x_0 \in \R^d$. 
The above result (the proof of which can be found in \cite{BGMP}) replaces Lemma \ref{lem:statphase}, valid in the free case. Note however, that one only infers 
a local result in some sufficiently small neighborhood of $x_0 \in \R^d$, provided the initial amplitude $a_0$ is sufficiently concentrated on \eqref{ypsilons}.
In order to obtain an estimate for $\eps \nabla r_\eps$, we note that by applying the Hamiltonian $H^\e$ to \eqref{sch}, and having in mind that $V\in L^\infty(\R^d)$, we infer 
\[\sup_{0<\e \le 1}\| \eps^2 \Delta \psi^\eps(t,\cdot) \|_{L^2}\le C, \quad \forall t \in \R_+,\]
where $C>0$ is independent of $\e$. In view of \eqref{multiWKB}, we consequently obtain that $ \| \eps^2 \Delta r_\eps \|_{L^2}$ is uniformly bounded w.r.t. $\e$ and hence 
we can interpolate
\[
\| \eps \nabla r_\eps\|^2_{L^2}  \le C \, \|r_\eps\|_{L^2}\,  \|\eps^2 \Delta r_\eps\|_{L^2} =\mathcal O(\eps),
\]
to obtain $\|\eps \nabla r_\eps\|_{L^2}=\mathcal O(\sqrt{\eps})=o(1)$, as required in Theorem \ref{thm:after}. In order to apply 
the latter we also require $\nabla S_j \not = \nabla S_k$ for $j\not = k \in \{ 1, \dots, N\}$. This follows, from similar arguments as has been done in the free case. 
Indeed, if the gradients were the same, by following backward the Hamiltonian flow we would get that the curves were starting from the same point, which is a contradiction.

Thus, after using appropriate localization arguments,
the multi-phase form \eqref{multiWKB} combined with Theorem \ref{thm:after} allows to infer the
{\it same} qualitative picture for the classical limit of Bohmian dynamics 
in the case $V\not =0$, as we showed above for the free case. Using the same notation as above, we can summarize our discussion as follows.
\begin{proposition}
Let $V \in C^\infty_{\rm b}(\R^d)$ and $S_0 $ satisfy Assumption \ref{hypSfree}. Let
 $(t_0,x_0)\in (\R_t \times \R^d_x) \setminus \mathscr C$,
and assume that $\{ y \in \R^d  :  |a_0(y)| >0 \}$
is a sufficiently small neighborhood of
$\{ Y_1(t_0, x_0), \dots, Y_N(t_0, x_0)\}$.
Then there exists a small neighborhood $\mathcal U\subset \R_t \times \R_x^d$ of $(t_0,x_0)$ such that
$\beta\neq w$ inside $\mathcal U \times \R^d_p$. In particular, the Bohmian trajectories $(X^\e,P^\e)$
do not converge locally in measure to the classical Hamiltonian flow.
\end{proposition}

\section{Numerical simulation of Bohmian trajectories}\label{sec:num}

In this section we shall numerically study the behavior of Bohmian trajectories, mainly in the regime $0< \e \ll 1$ and in particular in situations where caustics appear.
Let us remark that the numerical implementation of Bohmian trajectories is used in applications of quantum chemistry, cf. \cite{DDP, NeFre}.

\subsection{Description of the numerical method}
For the numerical tracking of Bohmian trajectories $(X^\e, P^\e)$ it is necessary to 
solve the system (\ref{bohmflow}) for a given solution \( \psi^{\e} (t,x)\) of the 
Schr\"odinger equation (\ref{sch}). To this end, we will always consider initial data $\psi_0^\e \in \mathcal S(\R^d)$,
i.e., 
rapidly decreasing functions. This allows to numerically approximate the solution 
\( \psi^{\e} \) through a truncated Fourier series in the spatial 
coordinates by choosing the 
computational domain $\Omega_{\rm com}$ sufficiently large, i.e., such that \( |\psi^{\e}| \) is 
smaller than machine precision at the $\partial \Omega_{\rm com}$ (we use double precision which is 
roughly equivalent to \( 10^{-16} \)). Thus the function can be 
periodically continued as a smooth function with maximal 
numerical precision. In our numerical examples, we shall concentrate on the case of $d=1$ spatial dimension. The \( x \)-dependence of
 \( \psi^{\e} \) is consequently treated with a discrete Fourier transformation 
realized via a Fast Fourier Transform (FFT) in Matlab. We thereby always 
choose the resolution large enough so that the modulus of the Fourier 
coefficients decreases to machine precision which is achieved in the 
studied examples for \( 2^{10} \) to \( 
2^{14} \) Fourier modes. This resolution enables high precision interpolation 
from \( x \) to \( X^\e \) (see below). 

For the time-integration of the Schr\"odinger equation we shall rely on a {\it time-splitting method}. 
The  basic idea underlying these splitting methods is the Trotter-Kato formula
\cite{TK}, i.e.,
\begin{equation}\label{e11}
 \lim_{n\rightarrow\infty}\left(e^{-tA/n}e^{-tB/n}\right)^{n}=e^{-t\left(A+B\right)}
\end{equation}
where $A$ and $B$ are certain unbounded linear operators, for details 
see \cite{Ka}.
In particular this includes the cases studied by Bagrinovskii and
Godunov in \cite{BG} and by Strang \cite{ST}.
The formula \eqref{e11} allows to solve an evolutionary 
equation 
$$\partial_{t} u=\left(A+B\right)u, \quad u|_{t=0} = u_0,$$ in the following form
\[
u(t)=e^{c_{1}\Delta tA}e^{d_{1}\Delta tB}e^{c_{2}\Delta tA}e^{d_{2}\Delta tB}\cdots e^{c_{k}\Delta tA} 
e^{d_{k}\Delta tB}u_0,
\]
where $(c_{1},\ldots,c_{k})$ and $(d_{1},\ldots,d_{k})$
are sets of real numbers that represent fractional time steps.  
In the numerical treatment of \eqref{sch} we shall use a second order Strang splitting, i.e., 
$c_{i}=d_{i}=1$ for all $i$ except for $c_{1}=d_{k}=1/2$. The Schr\"odinger equation is consequently split into the following system:
\[i\e 
\partial_{t}u+\frac{\e^{2}}{2}\partial_{xx}u=0,\qquad  i\e 
\partial_{t}u=V(x)u. \] The first equation can then be explicitly integrated in 
Fourier space, using two FFT's. The second equation can explicitly be solved (in physical space) in the form $$u(t,x) = e^{- it V(x)/\e} u_0.$$

Next, in order to solve the Bohmian equations of motion (\ref{bohmflow}) 
for a given \( \psi^{\e}(t,x) \), we need to interpolate between the coordinate \( 
x \), in which \( \psi^{\e} \) is given, and the coordinate \( 
X^\e \). For this we use that the \( x \)-dependence of $\psi^\e$ is treated by  
Fourier spectral methods. Thus we can apply the  
representation of \( \psi^{\e} \) in terms of truncated Fourier 
series not only at the collocation points for which  the formulae for 
the discrete Fourier transform hold, but at {\it general intermediate}
points. The main drawback is that for such points there is no FFT 
algorithm known 
and the transformation is thus computationally more expensive. But since we only need to track a limited 
number of trajectories \( X^\e \) and since this interpolation 
method is of high accuracy, our approach is more efficient than, say, a low 
order polynomial (\textit{spline}) interpolation (as used, e.g., in \cite{DDP}). In order to obtain the Bohmian momentum $P^\e$ 
we interpolate, $x \leftrightarrow X^\e$ within \( \psi^{\e}(t,x) \) 
and \( \partial_{x}\psi^{\e}(t,x) \), for fixed time \( t \in \R \). To this end, we note that the latter is of course determined in Fourier space. We 
consequently compute $P^\e$ through $$ P^\e(t, X^\e) = 
\e \im\left(\frac{\partial_x \psi^\e(t, X^\e) }{\psi^\e(t, X^\e)}\right).$$ We test the  
accuracy of the interpolation by comparing different numbers of 
Fourier modes for the solution of the Schr\"odinger equation for a 
given set of computed trajectories. Once machine precision is assured for 
\( \psi^{\e} \) (i.e., the modulus of the Fourier coefficients 
decreases below \( 10^{-12} \), in our case), the difference between 
different interpolates can be shown to be of the same order. Thus we 
can conclude that the spatial resolution of the trajectories is of the 
order of \( 10^{-12} \), much better than plotting accuracy.

The time integration of the first equation of the system (\ref{bohmflow}) is 
performed with an explicit scheme (here, we shall use a standard fourth order Runge-Kutta method). 
This allows to compute the right-hand side of this equation with the already known 
values for \( X^\e \) at the previous time step. Note that we compute 
the solution to the Schr\"odinger equation either exactly in time 
(if \( V(x)=0 \)) or with second order time splitting for each stage of 
the Runge-Kutta scheme (whenever \( V(x)\neq 0 \)). We shall test the accuracy of 
the time integration scheme by assuring that the difference of the 
numerical solution for \( N_{t} \) time steps to the solution 
for \( 2N_{t} \) time steps is smaller that \( 10^{-4} \) and thus 
much smaller than plotting accuracy. Typically we use \( N_{t}=10^{4} 
\). In addition the accuracy of the splitting scheme is tested as in \cite{etna} 
by tracing the 
numerically computed energy $E_{\rm num}^\e(t)$ which due to 
unavoidable numerical errors is indeed a function of time. In our examples, the relative conservation 
of $E_{\rm num}^\e(t)$ is ensured to better than \( 10^{-7} \) implying again 
an accuracy of more than \( 10^{-5} \).

\begin{remark}
    For efficiency reasons, the computation of the trajectories \( X^\e \) is done at the same time for all \( X^\e \) . Thus, 
    in principle, it could happen that the identification of the trajectories in the 
    examples below do not reflect the actual dynamics. By tracing also individual trajectories, i.e., 
    by computing just one \( X^\e \) per run, we nevertheless are able to ensure that this is not the case and that 
    the shown 
    trajectories are indeed the correct ones. In particular our numerical code 
    captures the physically imporant property that Bohmian trajectories {\it do 
    not cross}, see, e.g., \cite{DDP} (see also the proof of Theorem \ref{thm:free} (iii)). This is indeed a delicate issue in other numerical approaches 
    where the system (\ref{classflow}) is numerically integrated 
    with (\ref{hj}) and (\ref{wkbamp}) instead of (\ref{sch}), and 
    where  
    different interpolation techniques are used. The latter have to be chosen in a 
    way to avoid the crossing of the trajectories (see Section \ref{sec:vortices} below). 
\end{remark}

\subsection{Case studies} In the following we shall illustrate our analytical results by numerical examples, starting with the (globally smooth) case of semiclassical wave packets, which has 
already been treated in an earlier paper \cite{MPS2}. We shall then also consider the case of $\psi_0^\e$ exhibiting vortices before we finally deal with WKB initial data producing caustics.

\subsubsection{Vortices}\label{sec:vortices} 
Before studying the regime $0< \e \ll 1$ we shall show that our numerics displays an important non-crossing property of Bohmian trajectories. 
Indeed, it is well known that solutions to the Schr\"odinger equation \eqref{sch} in general can have nodes, i.e., points 
at which the wave function vanishes. Due to the superfluid property of $\psi^\e$ such nodes represent quantum mechanical {\it vortices}.
At such points, the Bohmian 
trajectories $X^\e$ are not well defined, but since \( P^\e \) does vanish as well 
at these points, there is a natural analytic continuation of the 
trajectories through such nodes. 
In the following, we shall numerically study the example given in \cite{BDGPZ}. 
More precisely, $\psi^\e$ is given by the 
superposition of the ground state and the second excited state of the harmonic oscillator (we also put \( \e=1 \) in this example),
i.e.,
\[ 
\psi(t,x)=\left(1+(1-2x^2)e^{-2it}\right)e^{-x^{2}-{it}/{2}}. \]
This wave function vanishes for \( x=0 \) and for all times \( t=(2k+1)\pi/2 \), with \( 
k\in\mathbb{Z} \). To treat the limit `0/0' numerically, we add some 
quantity of the order of the rounding error to the wave function 
which will consequently provide the limit with an error of the order of the 
unavoidable numerical error. The resulting trajectories can be seen 
in Fig.~\ref{vortex}. Note that indeed, all trajectories avoid the vortices at \( 
t=\pi/2 \) and \( t=3\pi/2 \), only the trajectory for \( x=0 \) 
passes through these nodes.
\begin{figure}[thb!]
\includegraphics[width=0.6\textwidth]{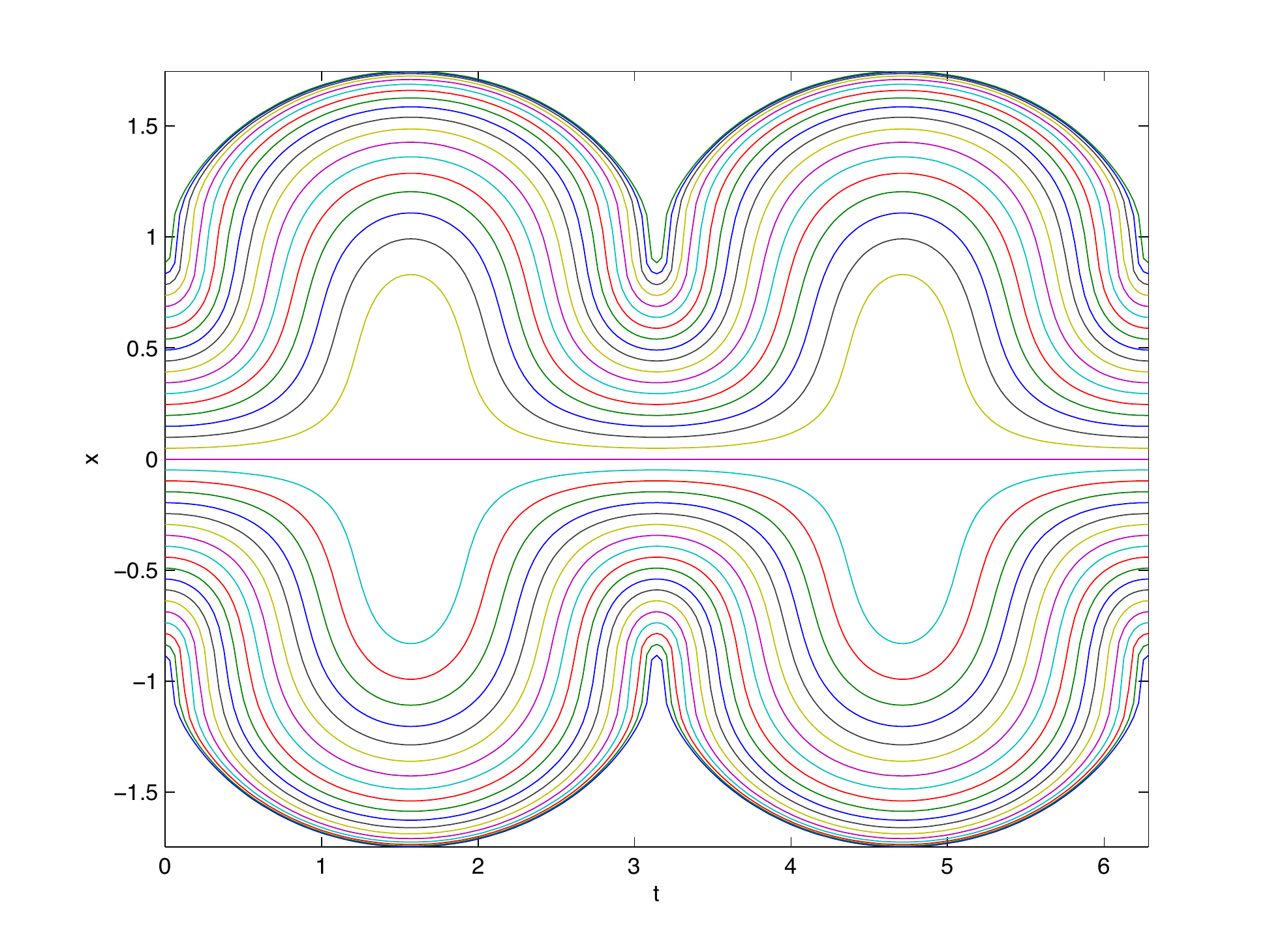}
 \caption{Bohmian trajectories for \( \e=1 \) in a harmonic oscillator potential 
 $V(x)=\frac{1}{2}x^{2}$ with $\psi$ given as a superposition of the ground state and 
the second excited state.}
 \label{vortex}
\end{figure}

\subsubsection{Semiclassical wave packets} 

In \cite{MPS2} a result similar to Theorem \ref{th1} (ii) is proved, for 
the case of {\it semiclassical wave-packets} (see also \cite{DuRo} for a closely related study). The corresponding initial data are of the form
\begin{equation}\label{wvp}
\psi^\e_0(x)= \e^{-d/4}\, a_0\left(\frac{x-x_0}{\sqrt{\e}}\right) e^{ik\cdot (x-x_0)/\e}, \quad a_0\in \mathcal S(\R^d;\C).
\end{equation}
The main differences between WKB states and semiclassical wave packets are that for the latter,
the particle density concentrates in a point, i.e.,
$$
\rho_0^\e(x) \stackrel{\e\rightarrow 0_+ }{\longrightarrow} \delta(x-x_0), \quad \text{in $\mathcal D'(\R^d)$,}
$$
and that the corresponding semiclassical approximation does {\it not exhibit caustics}, cf. \cite{CaFe} for more details. This in particular 
implies that for semiclassical wave packets one can 
prove convergence of the Bohmian trajectories on {\it any} finite time-interval \cite{MPS2}. 
An example for such a situation (with $k_0 = 0$) can be seen in  Fig.~\ref{coh}. 
\begin{figure}[thb!]
\includegraphics[width=0.5\textwidth]{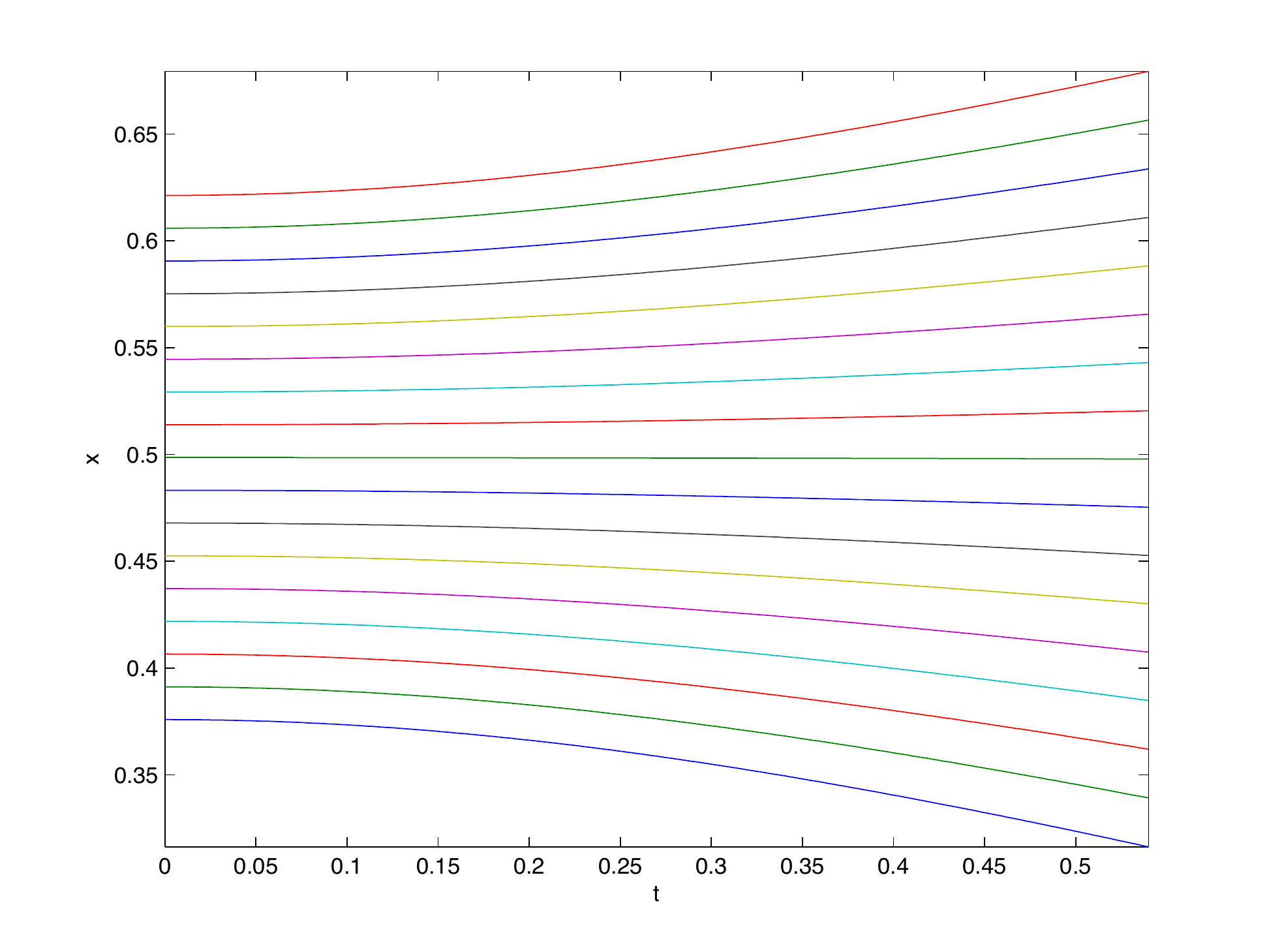}
 \caption{Bohmian trajectories for wave packet initial data of the form \eqref{wvp}
with $k_0 = 0$, $x_0 = 1/2$, $a_0(z) = e^{-z^2}$
and $\e=10^{-3}$.}
 \label{coh}
\end{figure}
The corresponding classical trajectories would be just lines parallel to the $t$-axis. 
Since these data do not lead to a caustic, there is just a slight 
defocusing effect to be seen with respect to the classical trajectories.

\subsubsection{Caustics}
In this last subsection we shall, finally, present examples exhibiting caustics in the 
classical limit. To this end, we shall first study the case where the caustic is 
just one single point, i.e., a situation in which all classical trajectories $X(t,y)$ cross at $(x^*, T^*)\in \R_t\times \R_x$. As 
a particular example, we shall consider the harmonic oscillator with potential 
\begin{equation}\label{hopot}V(x)=\frac{1}{2} \left(x-\frac{1}{2}\right)^{2},\end{equation}  and an initial data in the form
\begin{equation}\label{hoinin}\psi^\e_0(x)=e^{-25(x-1/2)^{2}},\\
\end{equation} 
i.e., a WKB state with Gaussian amplitude and $S_0 (x) =0$. Then, the classical
trajectories $X(t,y)$ all intersect in one point as can be seen in 
Fig.~\ref{X1e3HOu0newton}. 
\begin{figure}[thb!]
\includegraphics[width=0.5\textwidth]{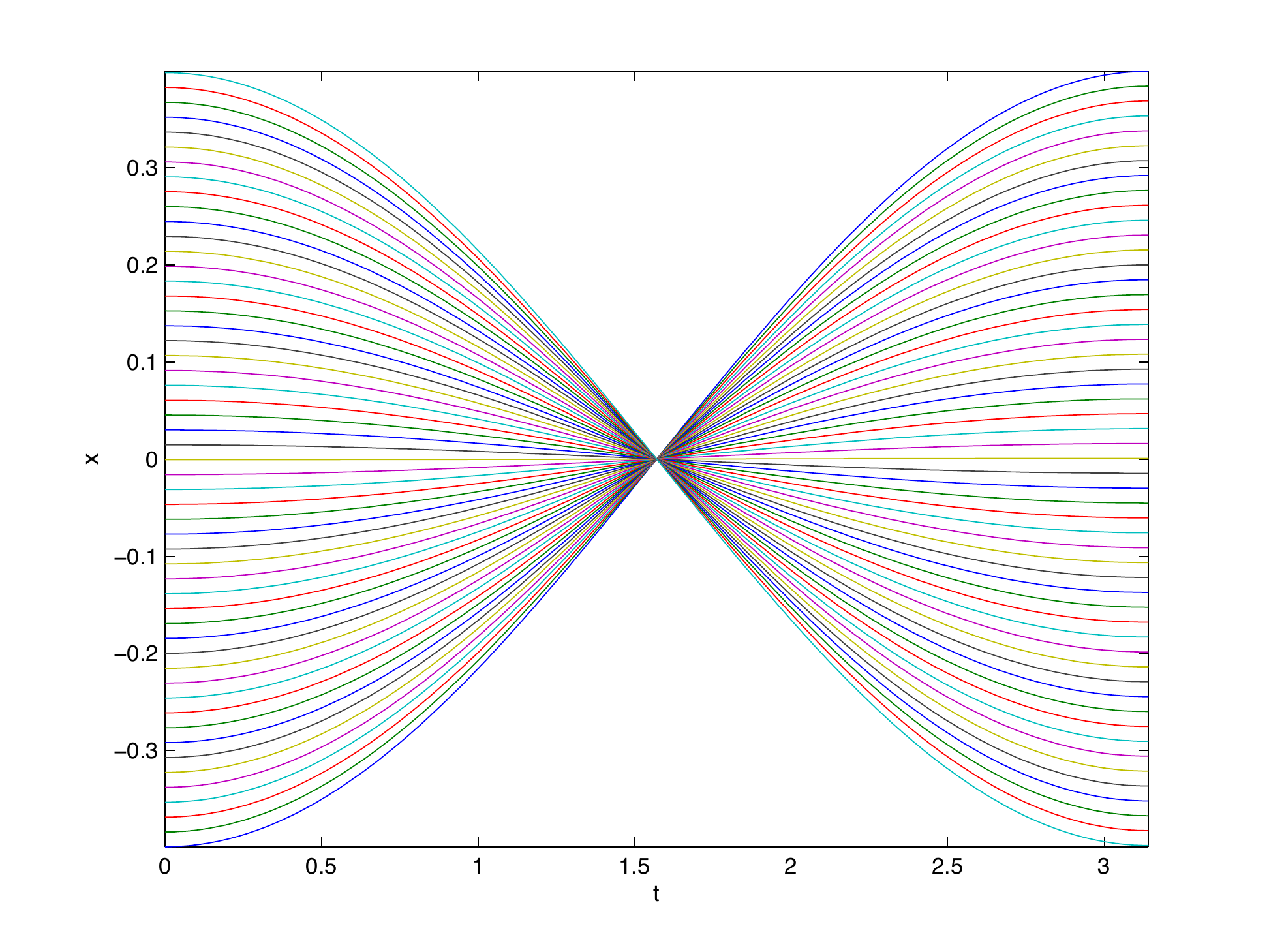}
 \caption{Classical trajectories $X(t,y)$ for the harmonic oscillator potential \eqref{hopot} and $\psi_0^\e$ given by \eqref{hoinin}.}
 \label{X1e3HOu0newton}
\end{figure}

The same situation for the Bohmian trajectories $X^\e(t,y)$ and $\e=10^{-3}$ can be seen in 
Fig.~\ref{X1e3HOu0}. The closeup of the  region of intersection when $\e=0$ 
clearly shows that the trajectories come close to $x^*$, but keep a finite distance from it except for the one trajectory 
which is parallel to the $t$-axis and goes straight through $x^*$. 
The solution $\psi^\e$ is periodic in time and shows  a 
breather-type behavior with a large $|\psi^\e|$ at the caustic. We show 
only a half-period of this periodic motion.
\begin{figure}[thb!]
    \includegraphics[width=0.45\textwidth]{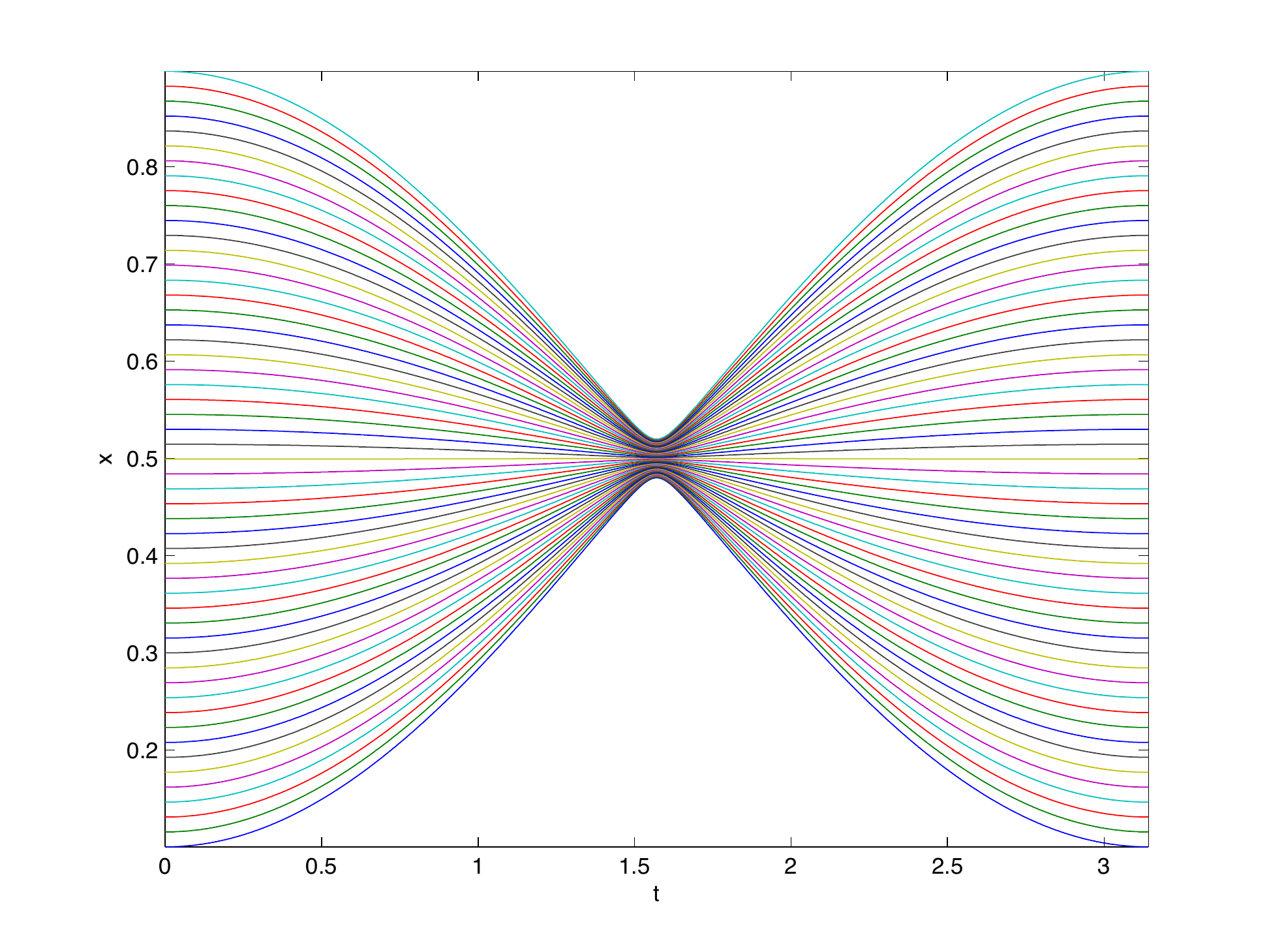}
    \includegraphics[width=0.45\textwidth]{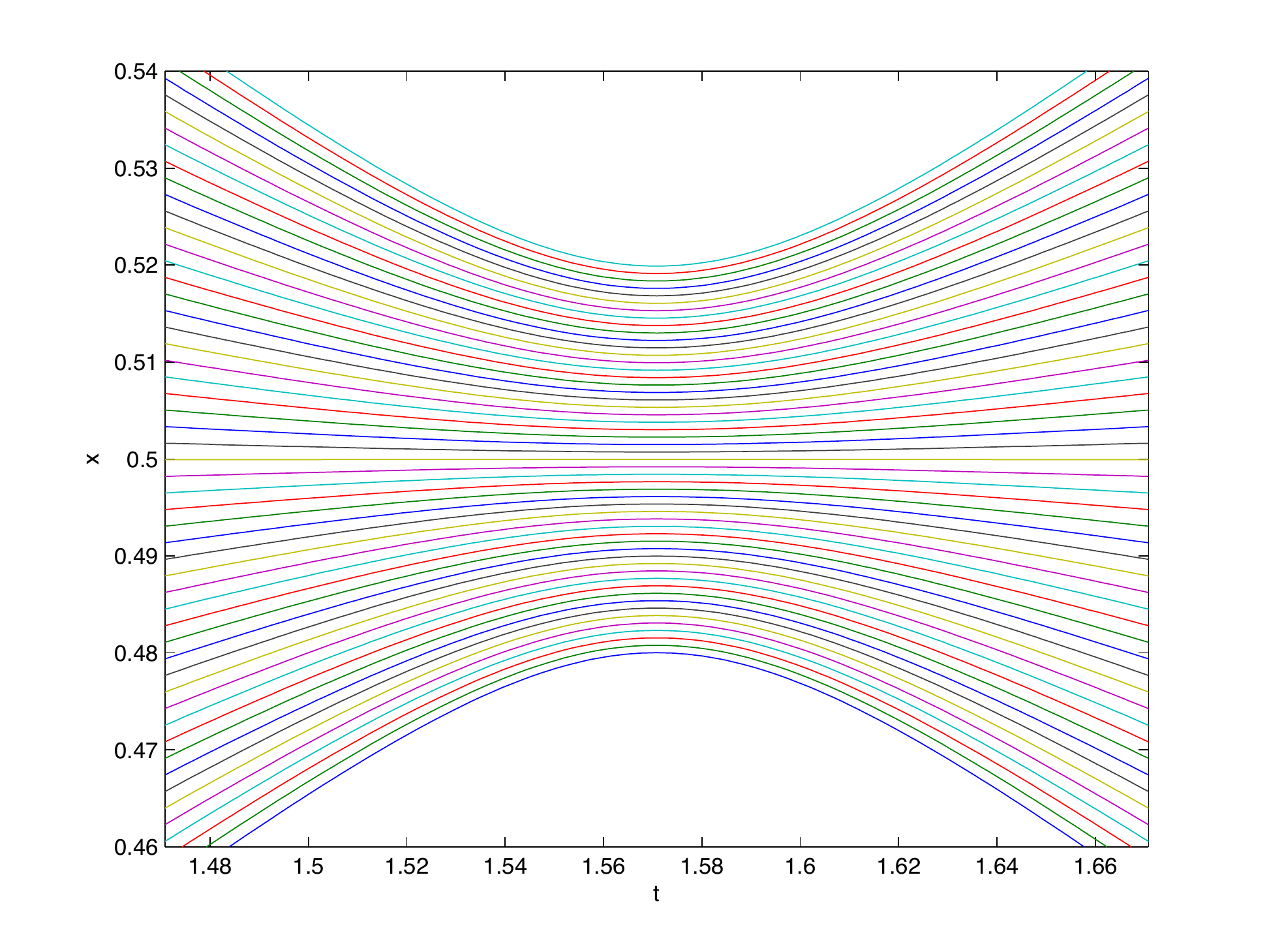}
 \caption{Left: Bohmian trajectories $X^\e(t,y)$ for the harmonic oscillator potential \eqref{hopot} 
 and $\psi_0^\e$ given by \eqref{hoinin}. Right: A closeup of the central region near $x^*$.}
 \label{X1e3HOu0}
\end{figure}

Next, we consider the case $V(x)=0$ with WKB initial data  
\begin{equation}\label{numfreeini}
\psi_0(x)=e^{-25(x-1/2)^{2}} e^{i S_0(x)/\e},\quad S_0(x) = 
-\frac{1}{5} \ln\cosh\left (5 x-\frac{5}{2}\right)
\end{equation} 
as in \cite{MPP},
i.e., the same amplitude as before but with nonzero initial phase. 
The time dependence of the density $\rho$ shows a strong maximum 
followed by a zone of oscillation inside a break-up zone as can be seen in 
Fig.~\ref{rho1e3}.
\begin{figure}[thb!]
    \includegraphics[width=\textwidth]{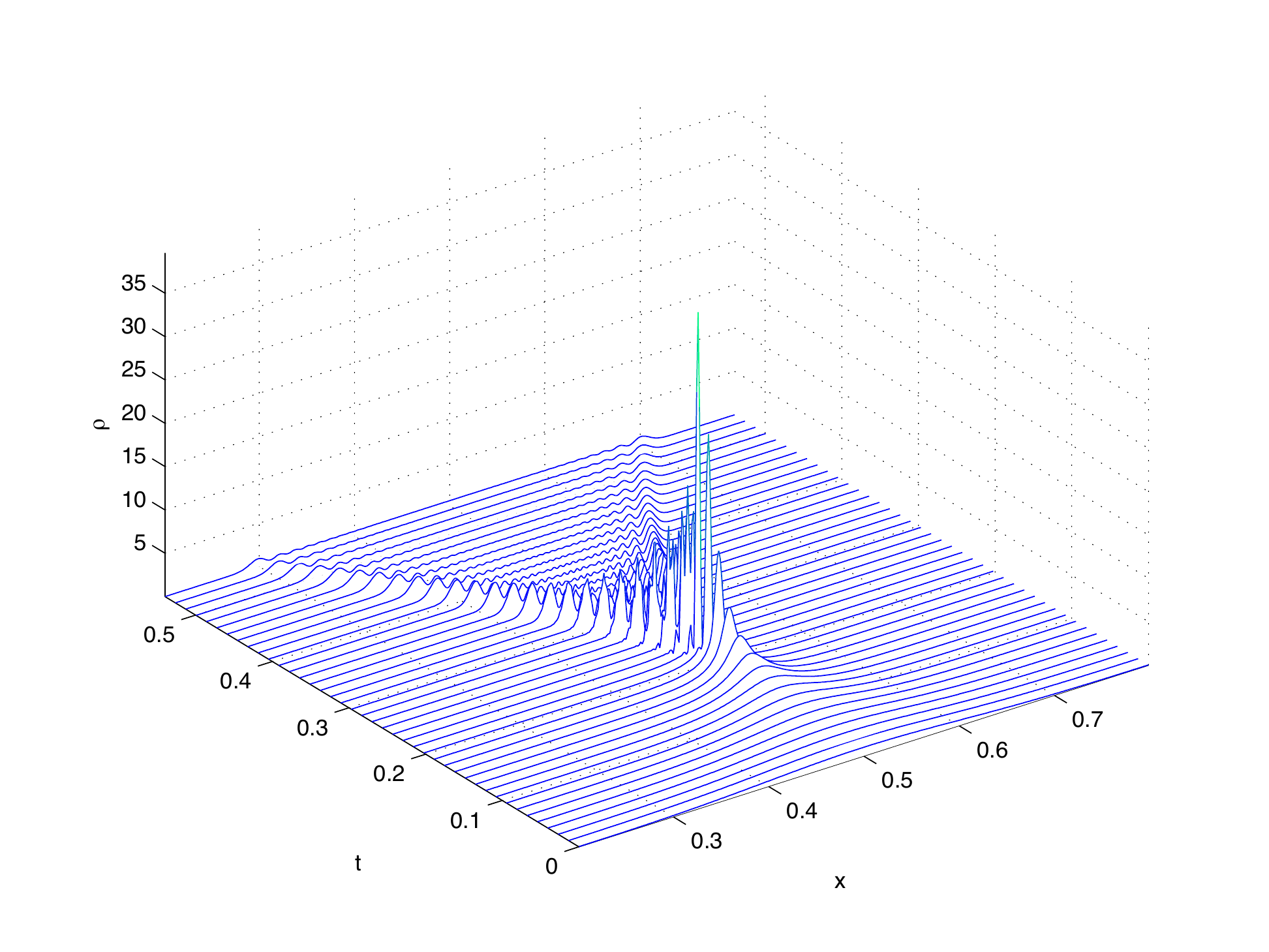}
 \caption{Density $\rho=|\psi^{\epsilon}|^{2}$ for $V(x)=0$ and $\psi_0^\e$ given by \eqref{numfreeini}.}
 \label{rho1e3}
\end{figure}
In this case, the classical trajectories $X(t,y)$ will lead to a diffuse caustic as depicted in 
Fig.~\ref{caustic}. For finite \( \e \),  
the Bohmian trajectories $X^\e(t,y)$ obviously do not cross, but there are rapid 
oscillations within the caustic 
region as can be seen in Fig.~\ref{Xcau1e3}.
\begin{figure}[thb!]
\includegraphics[width=0.45\textwidth]{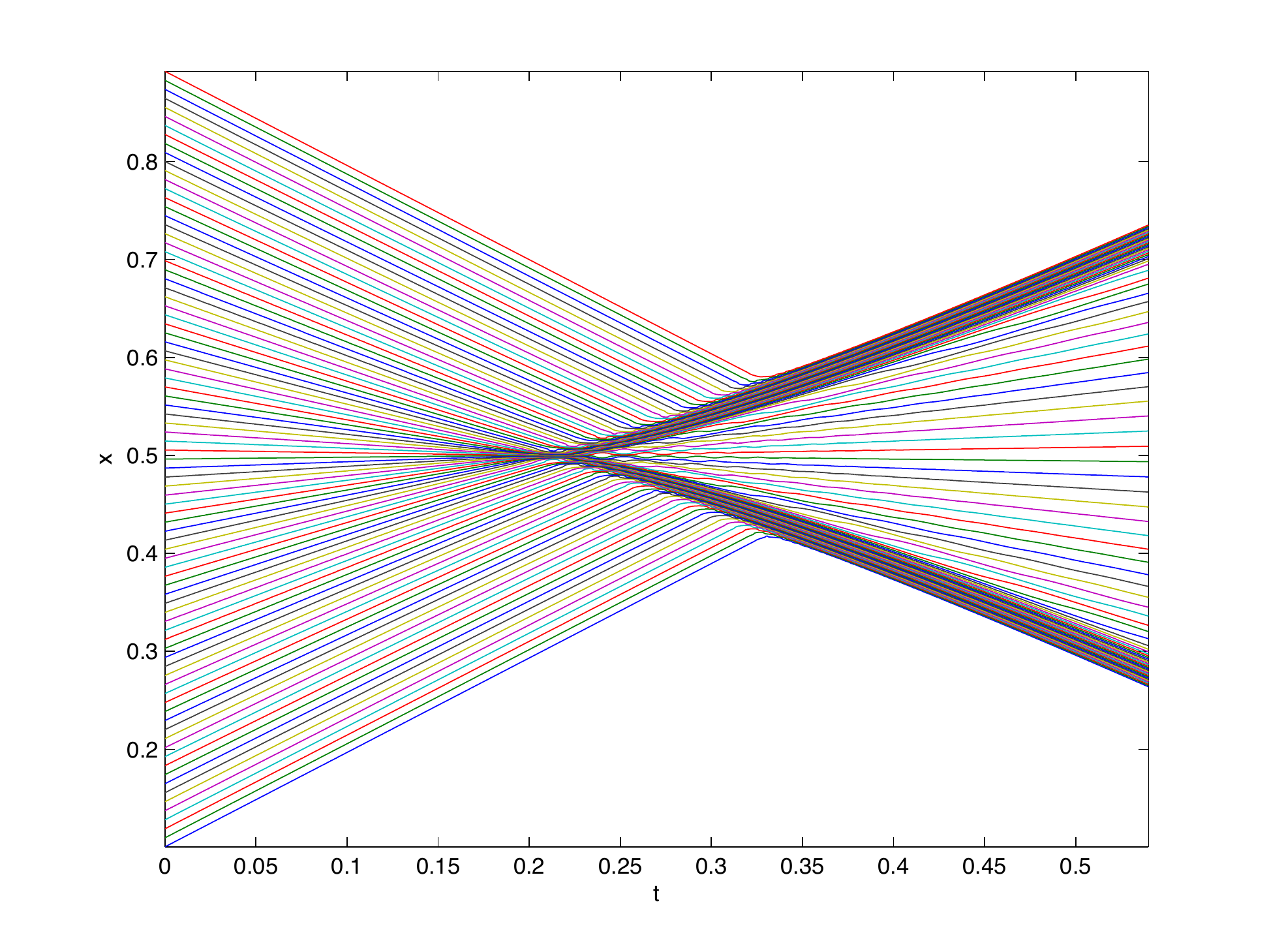}
\includegraphics[width=0.45\textwidth]{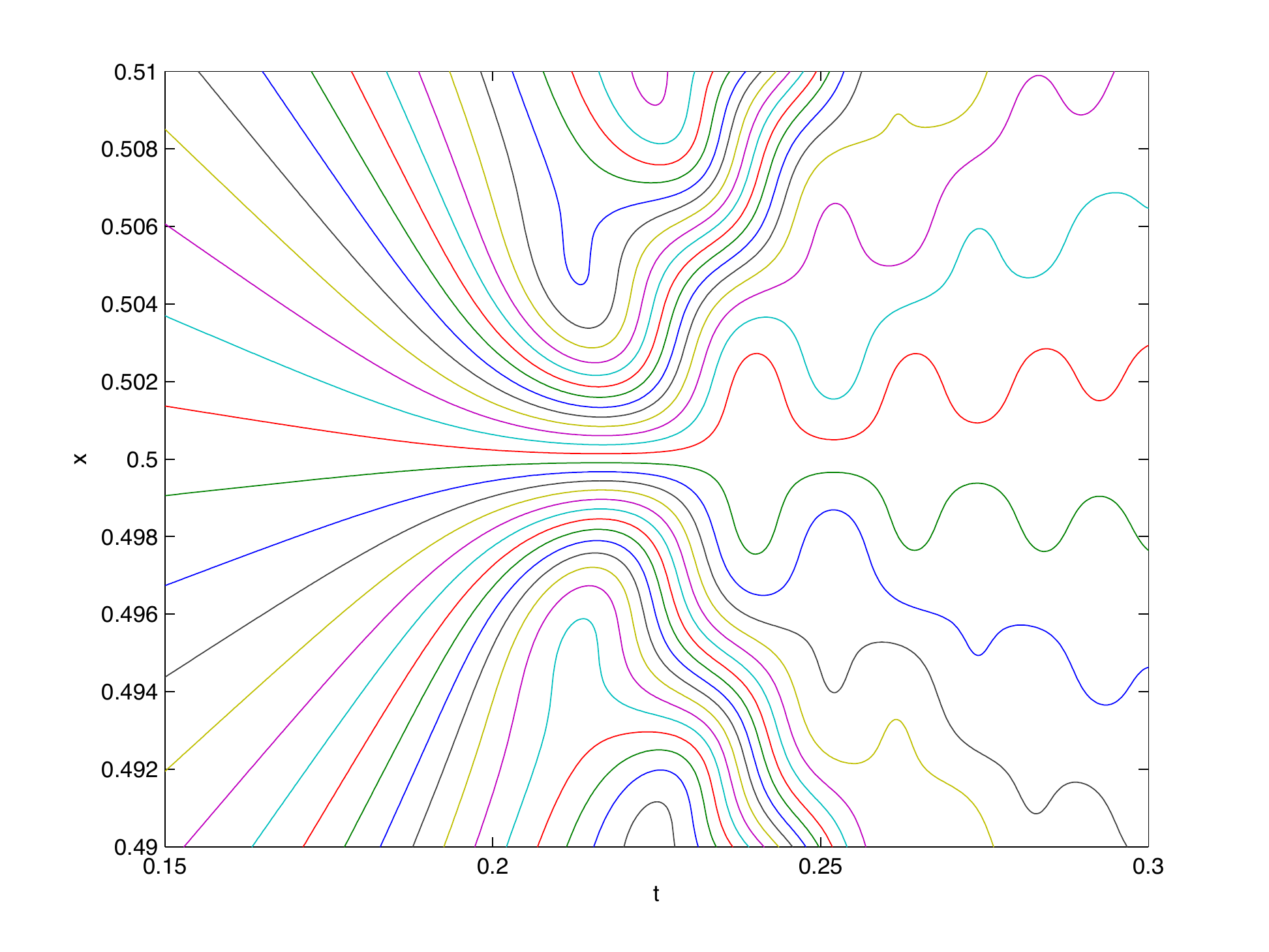}
 \caption{Left: Bohmian trajectories $X^\e(t,y)$ for $V(x)=0$ and $\psi_0^\e$ given by \eqref{numfreeini}. Right: A closeup of the central region.}
 \label{Xcau1e3}
\end{figure}

However, oscillations do not only appear in the trajectories, but 
also in the momentum \( P^\e(t,y) = u^\e(t,X^\e(t,y)) \) along any trajectory $X^\e$ which is ``deflected'' at the 
caustic region. This can be clearly seen in Fig.~\ref{Pcau1e3} where several
\( P^\e \) are plotted along the corresponding trajectories $X^\e$. The 
oscillations within $P^\e$ are reminiscent of so-called {\it 
dispersive shocks}, as observed, e.g., in the Korteweg-de Vries 
equation with small dispersion, see 
for instance \cite{GK} and references therein.
\begin{figure}[thb!]
\includegraphics[width=\textwidth]{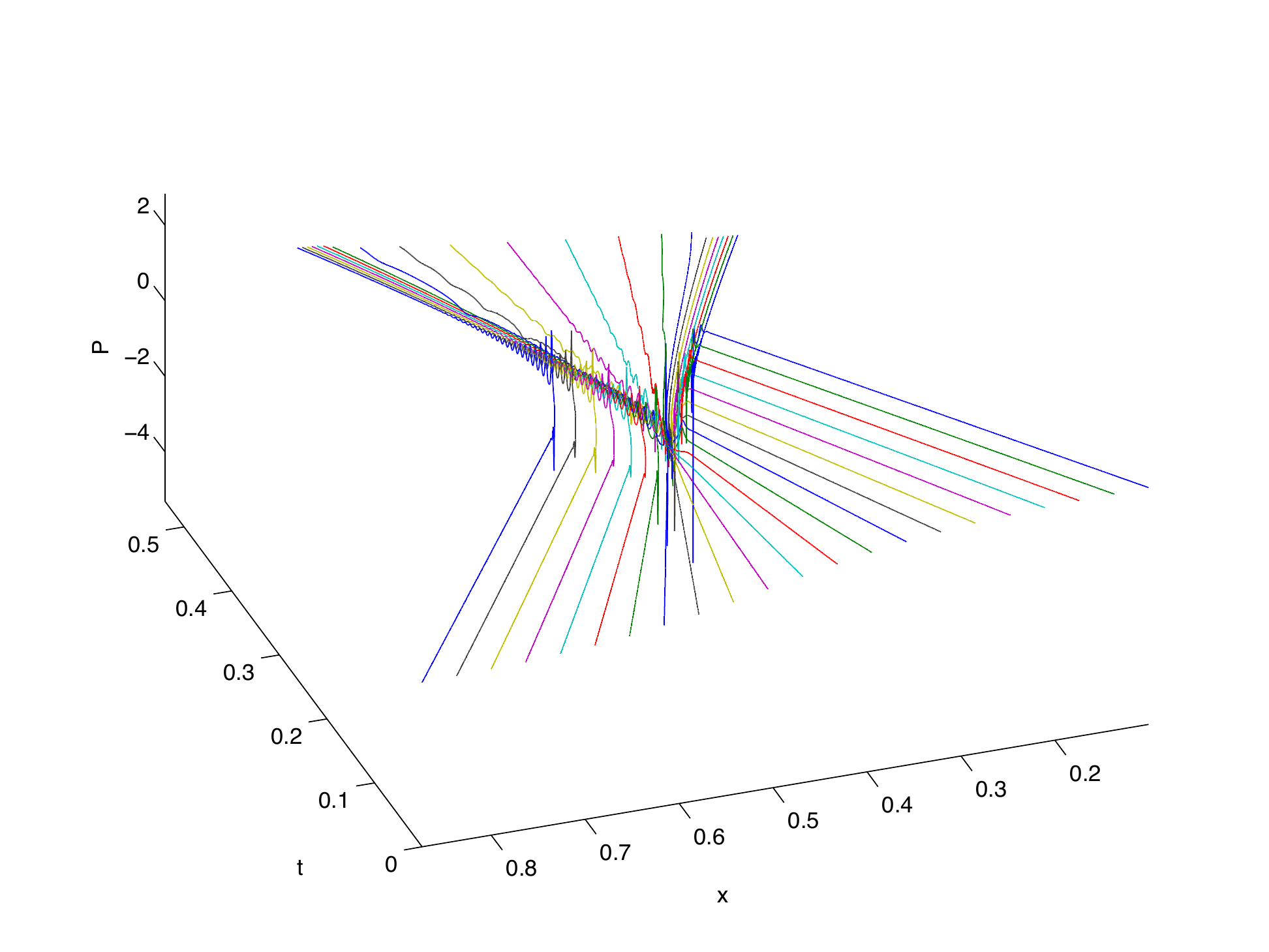}
 \caption{The quantity \( P^\e(t, y) \) along the Bohmian trajectories $X^\e(t,y)$ given 
 Fig.~\ref{Xcau1e3}.}
 \label{Pcau1e3}
\end{figure}
This is even more visible in Fig.~\ref{Pcau1e3single} where the 
oscillations on the left most trajectory in Fig.~\ref{Pcau1e3} are 
shown in dependence of $t$, thus in a projection onto the $t$-axis. 
\begin{figure}[thb!]
\includegraphics[width=0.6\textwidth]{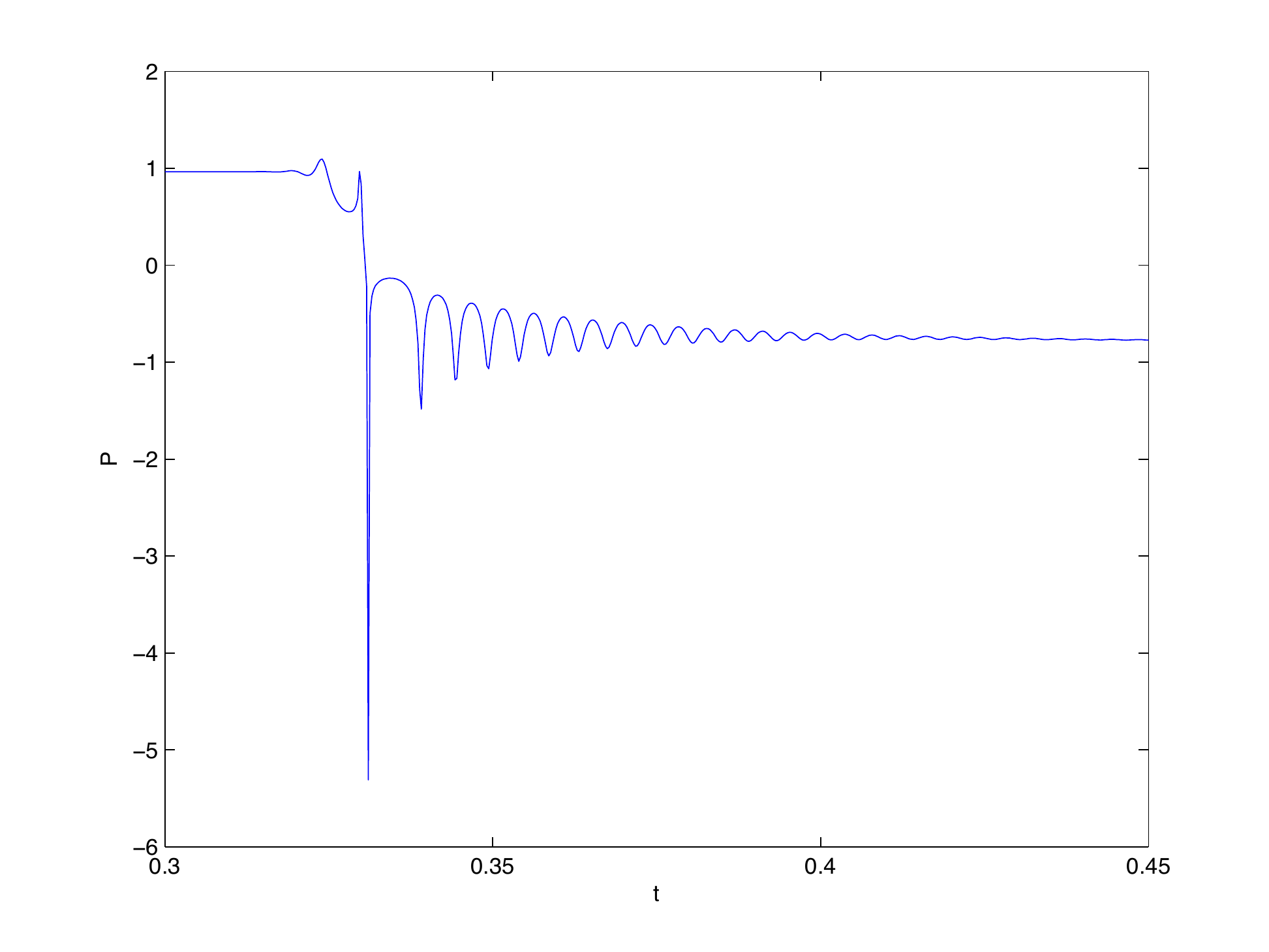}
 \caption{The quantity \( P^\e(t, y) \) along the left most 
 trajectory in Fig.~\ref{Pcau1e3} in dependence of $t$.}
 \label{Pcau1e3single}
\end{figure}

\end{document}